%% file: main.tex
\documentclass[runningheads]{llncs}
\input{header}

\usepackage[utf8]{inputenc}
\DeclareUnicodeCharacter{2028}{}
\usepackage[firstpage]{draftwatermark} \SetWatermarkText{\hspace*{4.3in}\raisebox{6.4in}{\includegraphics[scale=0.1]{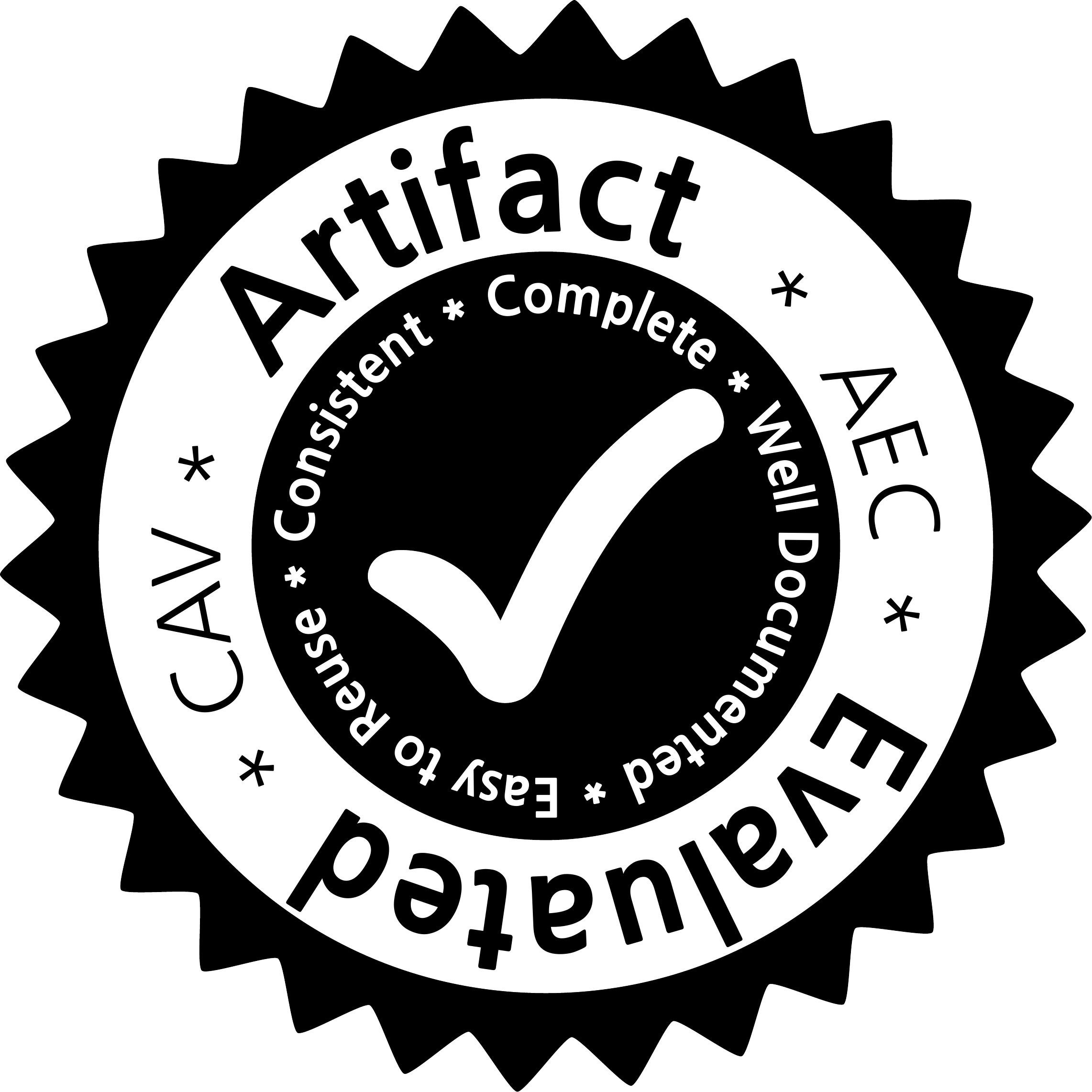}}} \SetWatermarkAngle{0}

\begin{document}

\title{Quantitative Mitigation of Timing Side Channels}
%
%

\author{
	Saeid Tizpaz-Niari
	\and
	Pavol {\v C}ern\'y
	\and
	Ashutosh Trivedi
}
\institute{University of Colorado Boulder}
\authorrunning{Tizpaz-Niari, {\v C}ern\'y, and Trivedi}
\maketitle
\begin{abstract}
  Timing side channels pose a significant threat to the security and privacy of
  software applications.  We propose an approach for {\em mitigating} this
  problem by decreasing the strength of the side channels as measured by
  entropy-based objectives, such as min-guess entropy. Our goal is to minimize
  the information leaks while guaranteeing a user-specified maximal acceptable
  performance overhead.
  We dub the decision version of this problem {\em Shannon mitigation}, and
  consider two variants, {\em deterministic} and {\em stochastic}. First, we
  show that the deterministic variant  is {\sc NP}-hard. However, we give a
  polynomial algorithm that finds an optimal solution from a restricted set.
  Second, for the stochastic variant, we develop an approach that uses
  optimization techniques specific to the entropy-based objective used. For
  instance, for min-guess entropy, we used mixed integer-linear programming.
  We apply the algorithm to a threat model where the attacker gets to make {\em
  functional observations}, that is, where she observes the running time of the
  program for the same secret value combined with different public input values.
  Existing mitigation approaches do not give confidentiality or
  performance guarantees for this threat model.
  We evaluate our tool \toolname on a number of micro-benchmarks and real-world
  applications with different entropy-based objectives. In contrast to the
  existing mitigation approaches, we show that in the functional-observation
  threat model, \toolname is scalable and able to maximize confidentiality under
  the performance overhead bound.
\end{abstract}

\input{intro}

\input{overview}

\section{Preliminaries}
\label{sec:def}
\input{definition}
\input{mitigation}

\input{experiments}

\input{case-study}

\input{related}

\input{ack}
\clearpage

\bibliography{papers}

\clearpage

\input{appendix}

\end{document}

%% file: header.tex
\usepackage{amsthm,paralist}
\usepackage{mdframed}
\usepackage{centernot}
\usepackage{comment}
\usepackage{mathtools,xparse}
\usepackage{pifont}
\usepackage{amsfonts}
\usepackage{url}

\usepackage{wrapfig}
\usepackage{wrapfig}
\usepackage{booktabs}
\usepackage[caption=false]{subfig}

\usepackage[ruled,vlined,linesnumbered,lined]{algorithm2e}

\usepackage{algpseudocode}
\usepackage{tikz}

\newcommand{\dis}{d}

\newcommand{\DIST}{\mathcal{D}}
\newcommand{\Ee}{\mathcal{E}}
\newcommand{\Ff}{\mathcal{F}}
\newcommand{\Bb}{\mathcal{B}}
\newcommand{\Gg}{\mathcal{G}}

\newcommand{\Pp}{\mathcal{P}}
\newcommand{\sPp}{\sem{\Pp}}

\newcommand{\Ss}{\mathcal{S}}

\newcommand{\sem}[1]{ [ \! [ {#1}  ]  \! ]} 
\def\rmdef{\stackrel{\mbox{\rm {\tiny def}}}{=}} 

\newcommand\vd[2]{d_{i, p}}
\newcommand\vy{\mathbf{y}}

\newcommand{\set}[1]{\left\{ #1 \right\}}

\newcommand{\seq}[1]{\langle #1 \rangle}

\newcommand{\R}{\mathbb R}

\newcommand{\Real}{\R}

\newcommand{\Rplus}{\R_{\geq 0}}
\newcommand{\SE}{\textsf{SE}}
\newcommand{\GE}{\textsf{GE}}
\newcommand{\MG}{\textsf{mGE}}
\newcommand{\toolname}{\textsc{Schmit}\xspace}
\theoremstyle{definition}

\newtheorem{exmp}{Example}[section]
\numberwithin{exmp}{section}

\definecolor{gold}{rgb}{0.99,0.78,0.07}
\usetikzlibrary{arrows,shapes,snakes,automata,backgrounds,positioning,decorations.pathmorphing,
	decorations.markings,calc}

\tikzstyle{dtreenode}=[draw=blue!10!gray,rounded rectangle, minimum size=5mm,fill=blue!10!white]
\tikzstyle{dtreeleaf}=[draw=black!60,minimum width=1cm,minimum height=0.4cm,rectangle,fill=blue!50!white]
\tikzset{every loop/.style={looseness=7}}
\tikzset{
	gluon/.style={decorate,draw=black,
		decoration={coil,amplitude=1pt, segment length=5pt}}
}
\tikzset{
	gluon1/.style={decorate,draw=black,
		decoration={coil,amplitude=3pt, segment length=3pt}}
}
\tikzset{
	gluonew/.style={decorate,draw=black,
		decoration={coil,amplitude=1pt, segment length=2pt}}
}

\tikzset{bicolor/.style args={#1 and #2 and #3}{
		path picture={
			\tikzset{rounded corners=0}
			\fill [#1] (path picture bounding box.south west)
			rectangle
			($(path picture  bounding box.north west)!#3!(path picture bounding
			box.north east)$);
			\fill [#2]
			($(path picture bounding box.south west)!#3!(path picture bounding
			box.south east)$)
			rectangle (path picture bounding box.north east);
}}}

\tikzset{tricolor/.style args={#1 and #2 and #3 and #4 and #5}{
		path picture={
			\tikzset{rounded corners=0}
			\fill [#1] (path picture bounding box.south west)
			rectangle
			($(path picture  bounding box.north west)!#4!(path picture bounding
			box.north east)$);
			\fill [#2]
			($(path picture bounding box.south west)!#4!(path picture bounding
			box.south east)$)
			rectangle
			($(path picture  bounding box.north west)!#5!(path picture bounding
			box.north east)$);
			\fill [#3]
			($(path picture bounding box.south west)!#5!(path picture bounding
			box.south east)$)
			rectangle (path picture bounding box.north east);
}}}

\usepackage{listings}

 \definecolor{dkgreen}{rgb}{0,0.6,0}
 \definecolor{gray}{rgb}{0.5,0.5,0.5}
 \definecolor{mauve}{rgb}{0.58,0,0.82}

\definecolor{cadmiumgreen}{rgb}{0.0, 0.42, 0.24}
\definecolor{verde}{rgb}{0.25,0.5,0.35}
\definecolor{jpurple}{rgb}{0.5,0,0.35}
\definecolor{darkgreen}{rgb}{0.0, 0.2, 0.13}

\newsavebox{\mybox}

%% file: intro.tex
\section{Introduction}
\label{sec:intro}
Information leaks through timing side channels remain a challenging problem
\cite{padlipsky1978limitations,lampson1973note,kocher1996timing,dhem1998practical,brumley2005remote,yarom2017cachebleed,phan2017synthesis}.
A program leaks secret information through timing side channels if an
attacker can deduce secret values (or their properties) by observing response
times. We consider the problem of mitigating timing side channels. Unlike
elimination techniques
\cite{agat2000transforming,molnar2005program,wu2018eliminating} that aim to
completely remove timing leaks without considering the performance penalty, the
goal of mitigation techniques
~\cite{kopf2009provably,askarov2010predictive,zhang2011predictive} is to weaken
the leaks, while keeping the penalty low.

We define the {\em Shannon mitigation} problem that decides whether there is a
mitigation policy to achieve a lower bound on a given security entropy-based
measure while respecting an upper bound on the performance overhead. Consider an
example where the program-under-analysis has a secret variable with seven
possible values, and has three different timing behaviors, each forming a
cluster of secret values. It takes $1$ second if the secret value is $1$, it takes $5$
seconds if the secret is between $2$ and $5$, and it takes $10$ seconds if the
secret value is $6$ or $7$. The {\em entropy-based measure} quantifies
the remaining uncertainty about the secret after timing observations. Min-guess
entropy~\cite{KB07,smith2009foundations,backes2009automatic} for this program
is $1$, because if the observed execution time is $1$,
the attacker guesses the secret in one try. A {\em mitigation policy} involves
merging some timing clusters by introducing delays. A good solution might be to
introduce a $9$ second delay if the secret is $1$, which merges two timing clusters.
But, this might be disallowed by the budget on the performance overhead. Therefore,
another solution must be found, such as introducing a $4$ seconds delay when the secret
is one.

We develop two variants of the Shannon mitigation problem: {\em deterministic} and
{\em stochastic}. The mitigation policy of the deterministic variant requires us
to move all secret values associated to an observation to another observation,
while the policy of the stochastic variant allows us to move only a portion of
secret values in an observation to another one. We show that the deterministic
variant of the Shannon mitigation problem is intractable and propose a dynamic
programming algorithm to approximate the optimal solution for the problem by
searching through a restricted set of solutions.
We develop an algorithm that reduces the problem in the stochastic variant
to a well-known optimization problem that depends on the entropy-based measure.
For instance, with min-guess entropy, the optimization problem is mixed integer-linear programming.

We consider a threat model where an attacker knows the public inputs
(known-message attacks~\cite{kopf2009provably}), and furthermore, where the public
input changes much more often than the secret inputs (for instance, secrets such as
bank account numbers do not change often). As a result, for each secret, the
attacker observes a timing function of the public inputs.
We call this model {\em functional observations} of timing side channels.

We develop our tool \toolname that has three components: side channel
discovery~\cite{FuncSideChan18}, search for the
mitigation policy, and the policy enforcement. The side channel discovery
builds the functional observations~\cite{FuncSideChan18}
and measures the entropy of secret set after the observations.
The mitigation policy component includes
the implementation of the dynamic programming and optimization algorithms.
The enforcement component is a monitoring system that uses the program
internals and functional observations to enforce the policy at runtime.

\noindent To summarize, we make the following contributions:
\begin{compactitem}
	\item We formalize the \emph{Shannon mitigation} problem with two
	variants and show that the complexity of finding deterministic
	mitigation policy is NP-hard.
	\item We describe two algorithms for synthesizing the mitigation policy:
	one is based on dynamic programming for the deterministic variant,
	that is in polynomial time and results in an approximate
	solution, and the other one solves the stochastic variant of the
	problem with optimization techniques.
	\item We consider a threat model that results in functional observations.
	On a set of micro-benchmarks, we show that
	existing mitigation techniques are not secure and efficient for this threat model.
	\item We evaluate our approach on five
	real-world Java applications. We show that \toolname
	is scalable in synthesizing mitigation policy within a few seconds and
	significantly improves the security (entropy) of the
	applications.
\end{compactitem}

%% file: overview.tex
\section{Overview}
\label{sec:overview}

First, we describe the threat model considered in this paper. Second, we
describe our approach on a running example. Third, we
compare the results of \toolname with the existing mitigation techniques
\cite{kopf2009provably,askarov2010predictive,zhang2011predictive} and show
that \toolname achieves the highest entropy (i.e., best mitigation) for all
three entropy objectives.

\noindent \textbf{Threat Model. }
\label{sec:threat}
We assume that the attacker has access to the source code
and the mitigation model, and she can sample the
run-time of the application arbitrarily many times on
her own machine. During an attack, she intends to guess a
fixed secret of the target machine by observing the mitigated
running time. Since we consider the attack models where the
attacker knows the public inputs and the secret inputs are less
volatile than public inputs, her
observations are functional observations, where for each secret value,
she learns a function from the public inputs to the running time.

\begin{figure*}[!t]
	\centering
	\begin{minipage}{0.5\textwidth}
		\begin{lstlisting}[frame=none]
		Example(int high, int low) {
		int t_high = high, t_low = low;
			while (t_high > 0) {
				if (t_high % 2 == 1) {
					while (t_low > 0) {
						if (t_low % 2 == 1) {
							res += compute(t_low,t_high);}
						t_low = t_low >> 1;}}
				t_high = t_high >> 1;}
			return res;}
		\end{lstlisting}
	\end{minipage}
	\hfill
	\begin{minipage}{0.45\textwidth}
	\includegraphics[width=\textwidth]{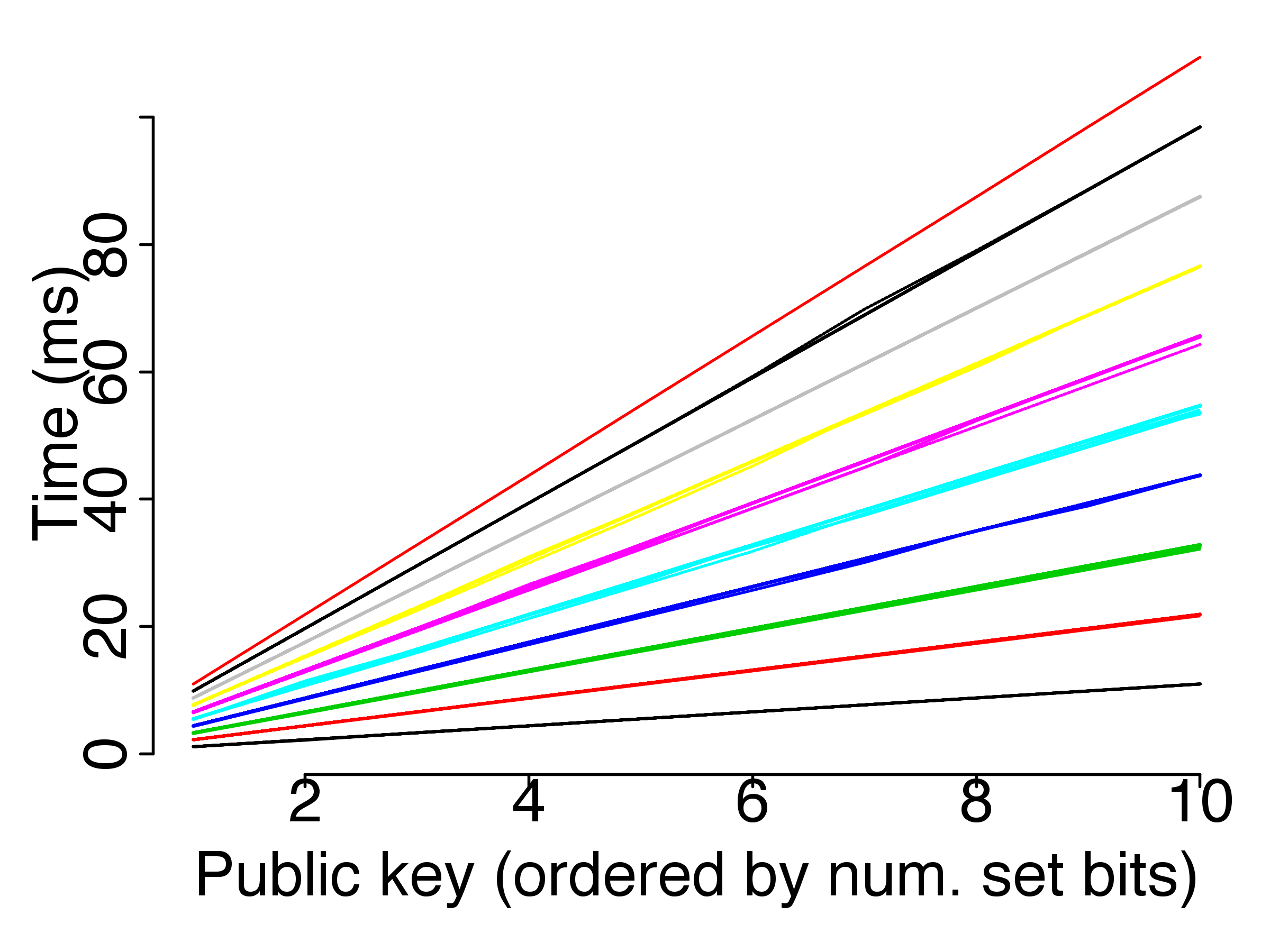}
	\end{minipage}
	\caption{ (a) The example used in Section~\ref{sec:overview}. (b) The timing functions for each secret value of the program.}
\label{fig:program-overview}
\end{figure*}

\begin{exmp}
Consider the program shown in Fig~\ref{fig:program-overview}(a).
It takes secret and public values as inputs. The running time depends on the number of set bits in both
secret and public inputs. We assume that secret and public inputs
can be between 1 and 1023.
Fig~\ref{fig:program-overview}(b) shows the running time of different
secret values as timing functions, i.e.,
functions from the public inputs to the running time.
\label{exp:overview}
\end{exmp}

\noindent \textbf{Side channel discovery.}
One can use existing tools to find the initial functional observations
~\cite{FuncSideChan18,aaai18}. In Example~\ref{exp:overview},
functional observations are $\Ff$ = $\seq{y, 2y$, $\ldots, 10y}$
where $y$ is a variable whose value is
the number of set bits in the public input.
The corresponding secret classes after this observation is
$\Ss_{\Ff} = \seq{1_1, 1_2, 1_3, \dots, 1_{10}}$
where $1_n$ shows a set of secret values that have $n$ set bits.
The sizes of classes are $B = \set{10,45,120,210,252,210,120,45,10,1}$.
We use $L_1$-norm as metric to calculate the distance between the
functional observations $\Ff$. This distance (penalty) matrix specifies
extra performance overhead to move from one functional observation to another.
With the assumption of uniform distributions over the secret input,
Shannon entropy, guessing entropy, and the
min-guessing entropy are 7.3, 90.1, and 1.0, respectively.
These entropies are defined in Section~\ref{sec:def} and
measure the remaining entropy
of the secret set after the observations. We aim to maximize
the entropy measures, while keeping the performance overhead
below a threshold, say 60\% for this example.

\noindent \textbf{Mitigation with \toolname.}
We use our tool \toolname to mitigate timing leaks of Example~\ref{exp:overview}.
The mitigation policy for the Shannon entropy
objective is shown in~Fig~\ref{fig:schmit}(a).
The policy results in two classes of observations.
The policy requires to move functional observations
$\seq{y, 2y, \ldots, 5y}$ to $\seq{6y}$ and all other
observations $\seq{7y, 8y, 9y}$ to $\seq{10y}$.
To enforce this policy, we use a monitoring system at runtime.
The monitoring system uses a decision tree model of the initial
functional observations. The decision tree model characterizes each functional
observation with associated program internals such as method
calls or basic block invocations~\cite{aaai18,tizpaz2017discriminating}.
The decision tree model for the Example~\ref{exp:overview} is shown
in Fig~\ref{fig:schmit}(b).
The monitoring system records program
internals and matches it with the decision tree model to
detect the current functional observation.
Then, it adds delays, if necessary,
to the execution time in order to enforce the mitigation policy.
With this method, the mitigated functional observation is
$\Gg$ = $\seq{6y, 10y}$ and the secret class is $\Ss_{\Gg} = \seq{\{1_1,1_2,1_3,1_4,1_5,1_6\},\{1_7,1_8,1_9,1_{10}\}}$
as shown in Fig~\ref{fig:schmit} (c). The performance overhead of
this mitigation is 43.1\%.
The Shannon, guessing, and min-guess entropies have improved
to 9.7, 459.6, and 193.5, respectively.

\begin{figure}[t!]
	\begin{minipage}{0.22\textwidth}
		\scalebox{0.5}{
			\begin{tikzpicture}[->,>=stealth',shorten >=1pt,auto,node distance=1.5cm,thick,
			draw = black!60, fill=black!60]
			\tikzstyle{every state}=[fill=blue!10!gray,draw=none,text=white]

			\node[state] 		 (A)                              {$C_{10}$};
			\node[state] 		 (B) [below of=A]         {$C_{x2}$};
			\node[state]         (C) [below of=B]         {$C_7$};
			\node[state]         (D) [below of=C]         {$C_6$};
			\node[state]         (E) [below of=D]          {$C_{x1}$};
			\node[state]         (F) [below of=E]         {$C_1$};

			\path (F) edge  [bend right=25]  node[midway,right] {\small $\mu(1,6) = 1.0$} (D)
			(E) edge node[midway, left] {\small $\mu(x1,6) = 1.0$} (D)
			(D) edge  [loop left]   node {\small $\mu(6,6) = 1.0$} (D)
			(C) edge [bend right=25] node[midway, right] {\small $\mu(7,10) = 1.0$} (A)
			(B) edge node[midway, left] {\small $\mu(x2,10) = 1.0$} (A)
			(A) edge [loop left]   node {\small $\mu(10,10) = 1.0$} (A);
			\end{tikzpicture}
		}
	\end{minipage}
	\begin{minipage}{0.35\textwidth}
		\scalebox{0.50}{
			\begin{tikzpicture}[align=center,node distance=0.8cm,->,thick,
			draw = black!60, fill=black!60]
			\centering
			\pgfsetarrowsend{latex}
			\pgfsetlinewidth{0.3ex}
			\pgfpathmoveto{\pgfpointorigin}

			\node[dtreenode,initial above,initial text={}] at (0,0) (l0)  {
				modExp\_bblock\_16};
			\node[dtreenode,below left=of l0] (l1)
			{modExp\_bblock\_16};
			\node[dtreenode,below right=of l0] (l2)
			{modExp\_bblock\_16};
			\node[dtreenode,below =of l1] (l3)
			{modExp\_bblock\_16};
			\node[dtreenode,below=of l2] (l6)
			{modExp\_bblock\_16};
			\node[below =of l3] (l7)
			{};
			\node[below =of l6] (l8)
			{};

			\node[dtreeleaf,bicolor={cyan and cyan and 0.99},below right=of
			l1] (l4) {};
			\node[dtreeleaf,bicolor={purple and purple and 0.99},below left=of l2]
			(l5) {};

			\path[->]  (l0) edge  node [left,pos=0.4] {$ <= 5.0 * y ~~$} (l1);
			\path  (l0) edge  node [right, pos=0.4] {$~~ > 5.0 * y $} (l2);
			\path  (l1) edge  node [left] {$~~ <= 4.0 * y $} (l3);
			\path  (l1) edge  node [right] {$~~ > 4.0 * y $} (l4);
			\path  (l2) edge  node [left] {$~~ <= 6.0 * y $} (l5);
			\path  (l2) edge  node [right] {$~~ > 6.0 * y $} (l6);
			\path  (l3) edge[dotted]  node [right] {} (l7);
			\path  (l6) edge[dotted]  node [right] {} (l8);
			\end{tikzpicture}
		}
	\end{minipage}
	\hfil
	\begin{minipage}{0.32\textwidth}
		\includegraphics[width=\textwidth]{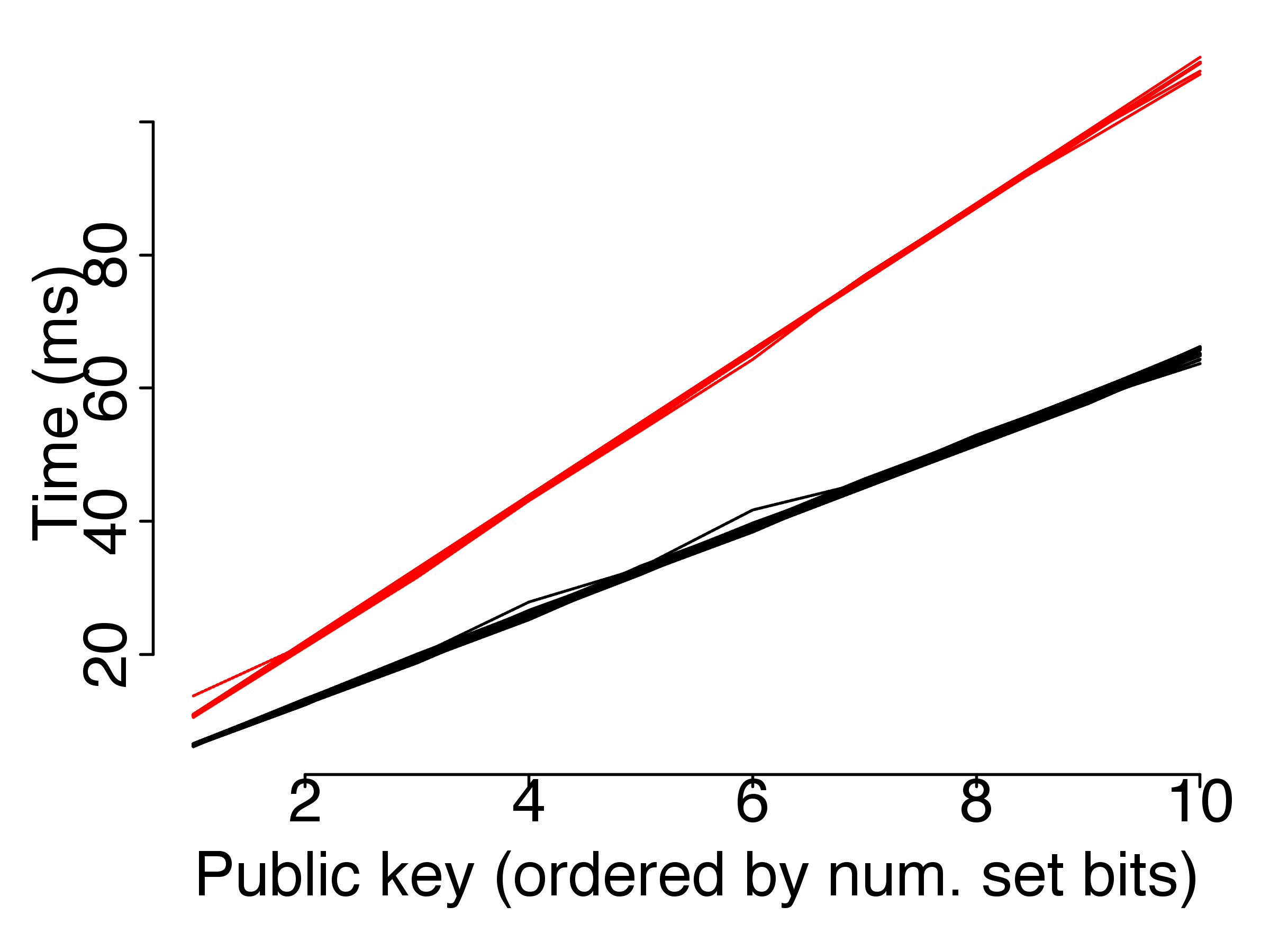}
	\end{minipage}
	\caption{
		(a)Mitigation policy calculation with deterministic
		algorithm (left). The observations $x1$ and $x2$ stands for all observations from $C_2{-}C_5$ and from $C_8{-}C_9$, resp.;
		(b) Leaned discriminant decision tree (center): it characterizes the functional clusters of Fig.~\ref{fig:program-overview}(b)
		with internals of the program in Fig.~\ref{fig:program-overview}(a); and
		(c) observations (right) after the mitigation by \toolname results in two classes of observations.
	}
	\label{fig:schmit}
\end{figure}
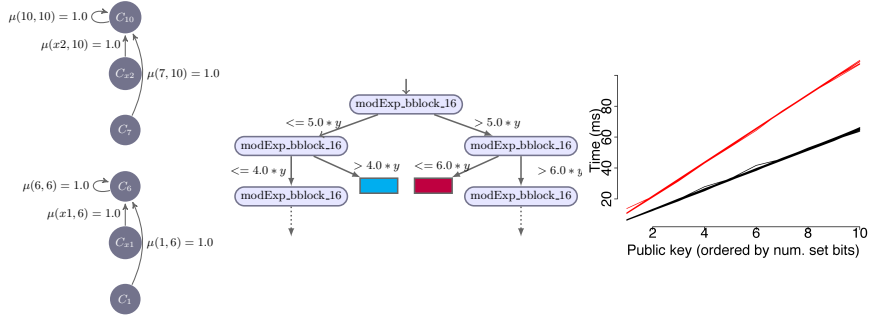

\noindent \textbf{Comparison with state of the art.}
We compare our mitigation results to black-box mitigation
scheme~\cite{askarov2010predictive} and
bucketing~\cite{kopf2009provably}.
\textit{Black-box double scheme technique.}
We use the double scheme technique~\cite{askarov2010predictive}
to mitigate the leaks of Example~\ref{exp:overview}.
This mitigation uses a prediction model to release
events at scheduled times.
Let us consider the prediction for releasing the event $i$ at $N$-th
epoch with $S(N,i)$ = $\max(inp_i, S(N,i{-}1)) {+} p(N)$,
where $inp_i$ is the time arrival of the $i$-th request, $S(N,i-1)$
is the prediction for the request $i{-}1$, and $p(N) = 2^{N-1}$
models the basis for the prediction scheme at $N$-th epoch.
We assume that the request are the same type and
the sequence of public input requests for each secret are
received in the beginnig of epoch $N=1$.
Fig~\ref{fig:black-box}(a) shows the functional observations
after applying the predictive mitigation.
With this mitigation, the classes of
observations are
$\Ss_{\Gg} = \seq{1_1,\{1_2,1_3\},\{1_4, 1_5,1_6,1_7\},\{1_8, 1_9,1_{10}\}}$.
The number of classes of observations
is reduced from 10 to 4. The performance overhead is 39.9\%.
The Shannon, guessing, and min-guess entropies have increased to
9.00, 321.5, and 5.5, respectively.
\textit{Bucketing.}
We consider the mitigation approach with buckets~\cite{kopf2009provably}.
For Example~\ref{exp:overview}, if the attacker does not know the public input
(unknown-message attacks~\cite{kopf2009provably}), the
observations are $\{1.1, 2.1, 3.3,\cdots, 9.9,10.9,\cdots,109.5\}$ as shown
in Fig~\ref{fig:bucket}(b). We apply the bucketing algorithm in~\cite{kopf2009provably}
for this observations, and it finds two buckets $\{37.5, 109.5\}$ shown with
the red lines in~Fig~\ref{fig:bucket}(b). The bucketing mitigation requires to move the
observations to the closet bucket. Without functional observations, there are 2
classes of observations. However, with functional observations,
there are more than 2 observations. Fig~\ref{fig:bucket}(c) shows
how the pattern of observations are leaking through functional side channels.
There are 7 classes of observations:
$\Ss_{\Gg} = \seq{\{1_1,1_2,1_3\},\{1_4\},\{1_5\},\{1_6\},\{1_7\},\{1_8\},\{1_9\},\{1_{10}\}}$.
The Shannon, guessing, and min-guess
entropies are 7.63, 102.3, and 1.0, respectively.
Overall, \toolname achieves the higher entropy measures for
all three objectives under the performance overhead of 60\%.

\begin{figure}[t!]
	\includegraphics[width=0.32\textwidth]{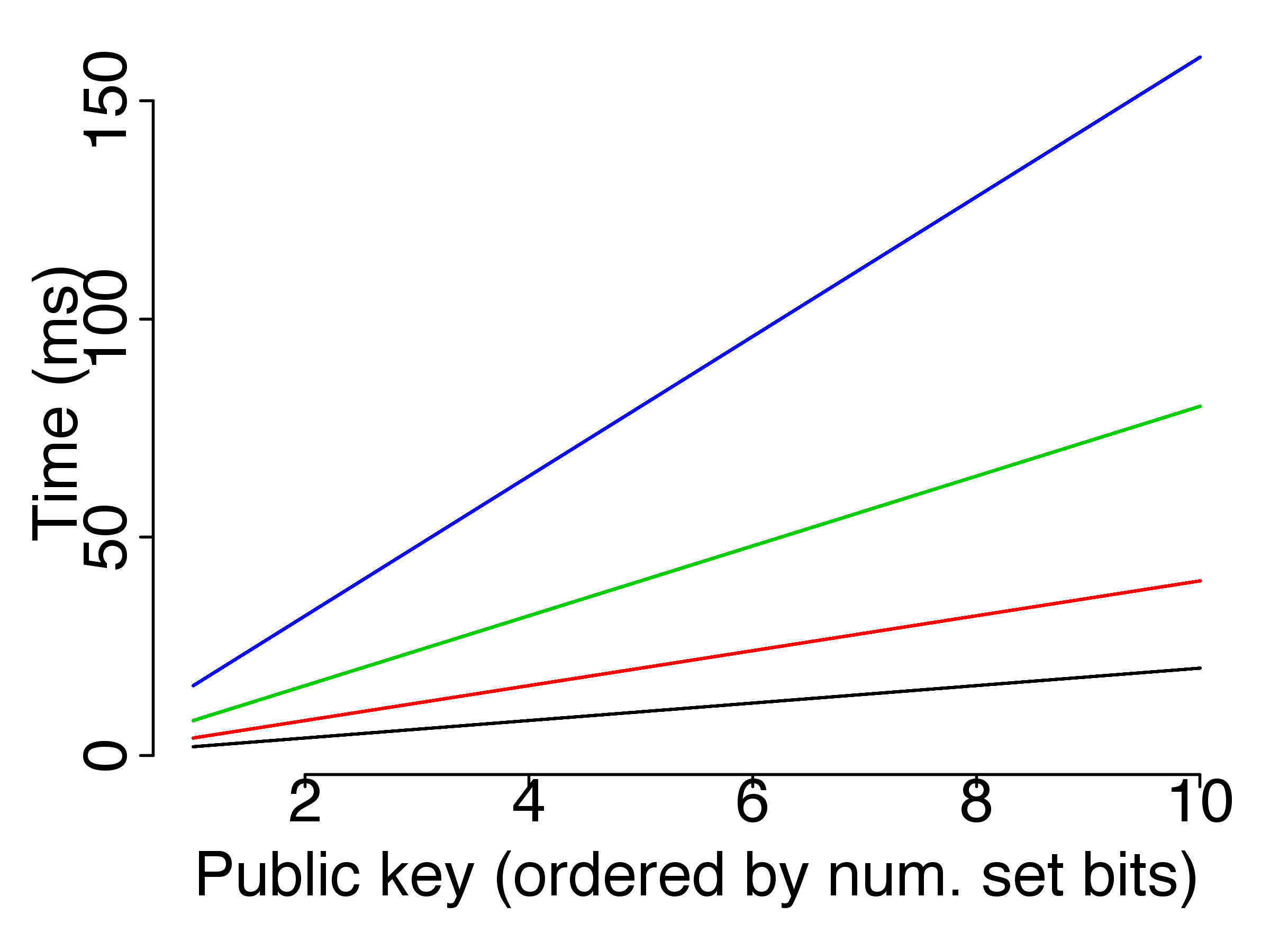}
	\includegraphics[width=0.32\textwidth]{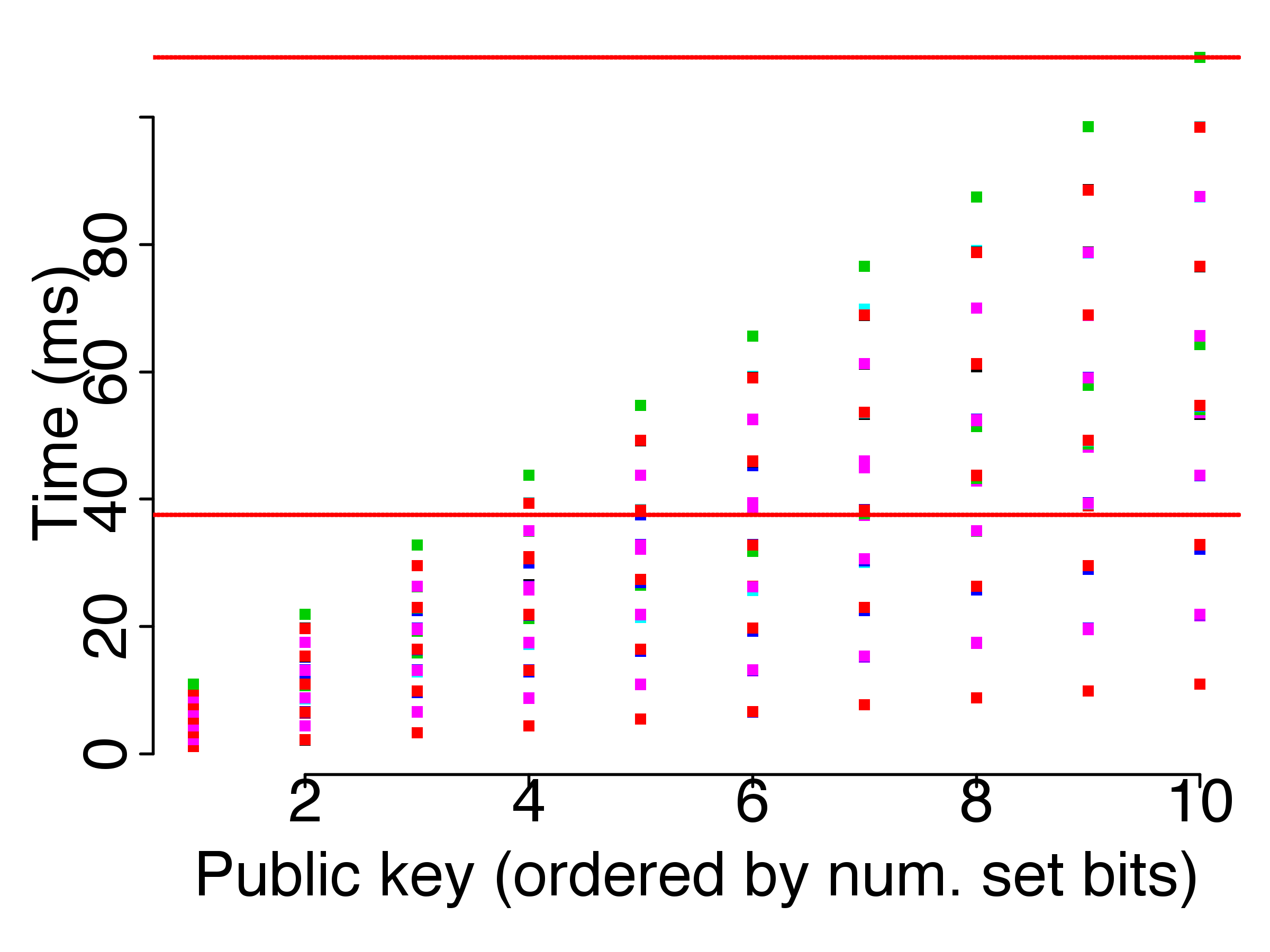}
	\includegraphics[width=0.32\textwidth]{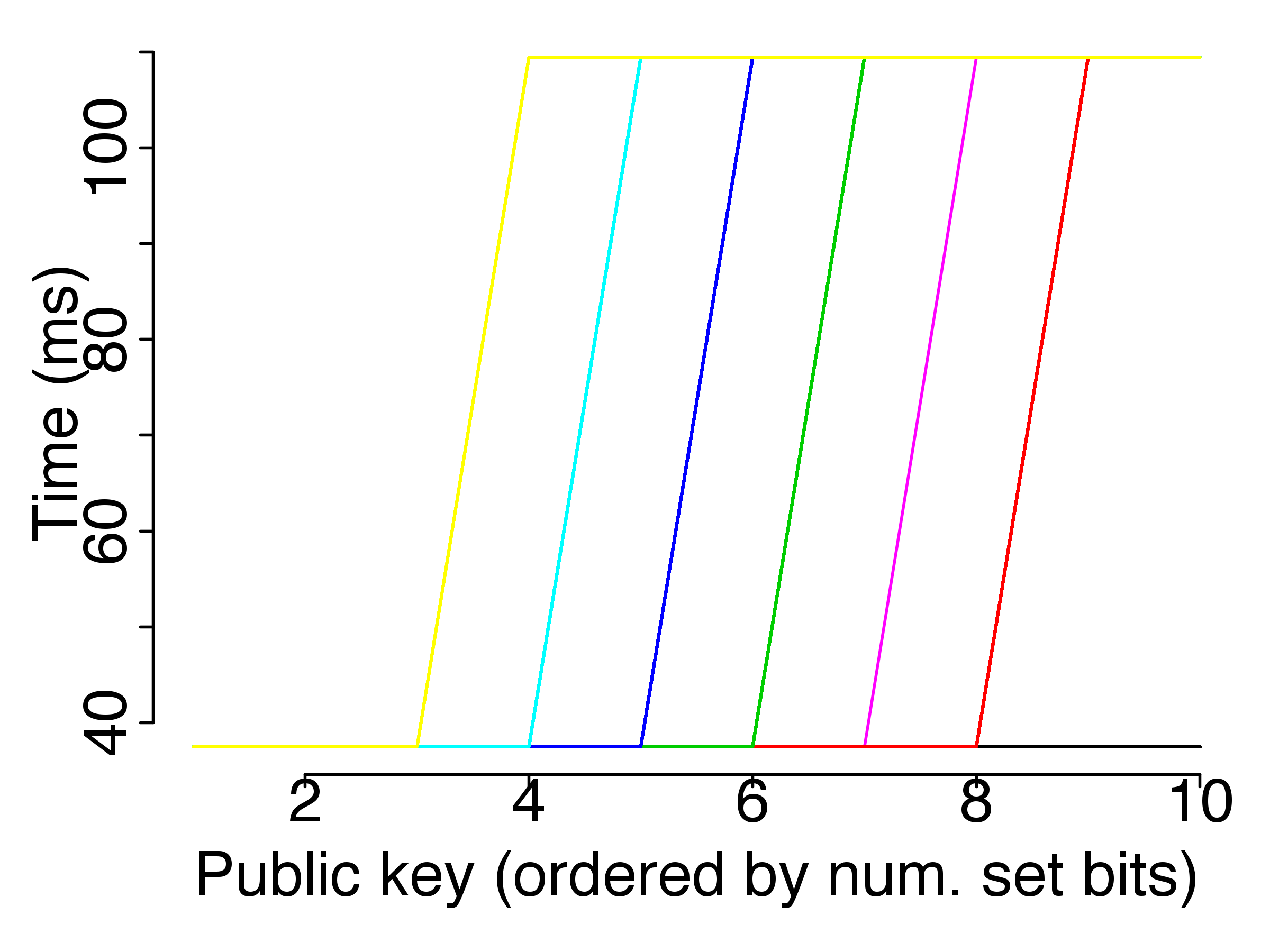}
	\caption{(a) The execution time after mitigation using the
		double scheme technique~\cite{askarov2010predictive}.
		There are four classes of functional observations
		after the mitigation.
		(b) Mitigation with bucketing~\cite{kopf2009provably}.
		All observations require to move to the closet red line.
		(c) Functional observations distinguish 7 classes of
		observations after mitigating with bucketing.}
	\label{tab:black-box}
	\label{fig:black-box}
	\label{fig:bucket}
\end{figure}

%% file: definition.tex
For a finite set $Q$, we use $|Q|$ for its cardinality.
A \emph{discrete probability distribution}, or just distribution, over a
set $Q$ is  a function $\dis : Q {\to} [0, 1]$ such that
$\sum_{q \in Q} \dis(q) = 1$.
Let $\DIST(Q)$ denote the set of all discrete distributions over $Q$.
We say a distribution ${\dis \in \DIST(Q)}$ is a \emph{point distribution}
if $\dis(q) {=} 1$ for a $q \in Q$.
Similarly, a distribution ${\dis \in \DIST(Q)}$ is \emph{uniform}
if $\dis(q) {=} 1/|Q|$ for all $q \in Q$.

\begin{definition}[Timing Model]
  The {\it timing model }of a program $\Pp$ is a tuple
  $\sPp = (X, Y, \Ss, \delta)$ where
  $X = \set{x_1, \ldots, x_n}$ is the set of {\it secret-input} variables,
  $Y = \set{y_1, \ldots, y_m}$ is the set of {\it public-input}
  variables,
  $\Ss \subseteq \Real^n$ is a finite set of {\it
  secret-inputs},  and $\delta: \Real^n \times \Real^m \to \Rplus$ is the
  execution-time function of the program over the secret and public inputs.
\end{definition}

We assume that the adversary knows the program and wishes to learn the value of
the secret input.
To do so, for some fixed secret value $s \in \Ss$, the adversary can invoke the
program to estimate (to an arbitrary precision) the execution time of the program.
If the set of public inputs is empty, i.e. $m = 0$, the adversary can only make
{\it scalar observations} of the execution time corresponding to a secret value.
In the more general setting, however, the adversary can arrange his
observations in a functional form by estimating an approximation of the {\it
  timing function $\delta(s) :  \Real^m \to \Rplus$} of the program.

A {\it functional observation} of the program $\Pp$ for a secret
input $s \in \Ss$ is the function $\delta(s): \Real^m \to \Rplus$ defined as
$\vy \in \Real^m \mapsto \delta(s, \vy)$.
Let $\Ff \subseteq [\Real^m \to \Rplus]$ be the finite set of all functional
observations of the program $\Pp$.
We define an order $\prec$ over the functional observations $\Ff$: for $f, g \in
\Ff$  we say that $f  \prec g$ if $f(y) \leq g(y)$ for all $y \in \Real^m$.

The set $\Ff$ characterizes an equivalence relation $\equiv_{\Ff}$,
namely secrets with equivalent functional observations, over the set $\Ss$,
defined as  following: $s \equiv_{\Ff} s'$ if there is an $f \in \Ff$ such that
$\delta(s) = \delta(s') = f$.
Let $\Ss_\Ff = \seq{S_1, S_2, \ldots, S_k}$ be the quotient space of $\Ss$
characterized by the observations $\Ff = \seq{f_1, f_2, \ldots, f_k}$.
We write $\Ss_{f}$ for the secret set $S \in \Ss_\Ff$ corresponding to the
observations $f \in \Ff$.
Let $\Bb = \seq{B_1, B_2, \ldots, B_k}$ be the size of observational
equivalence class in $\Ss_\Ff$, i.e. $B_i = |\Ss_{f_i}|$ for $f_i \in \Ff$ and let $B =
|\Ss| = \sum_{i=1}^k B_i$.

Shannon entropy, guessing entropy, and min-guess entropy are three prevalent
information metrics to quantify information leaks in programs.
K\"opf and Basin~\cite{KB07} characterize expressions for various
information-theoretic measures on information leaks when there is a uniform
distribution on $\Ss$ given below.
\begin{proposition}[K\"opf and Basin~\cite{KB07}]
\label{prop:overapprox}
Let $\Ff = \seq{f_1, \ldots, f_k}$ be a set of observations and let $\Ss$ be the set of secret values.
Let $\Bb = \seq{B_1, \ldots, B_k}$ be the corresponding size
of secret set in each class of observation and $B = \sum_{i=1}^k B_i$.
Assuming a uniform distribution on $\Ss$, entropies can be characterized as:
\begin{enumerate}
\item {\bf Shannon Entropy:}
  $\SE(\Ss|\Ff) \rmdef (\frac{1}{B})  \sum_{1 \leq i \leq k} B_i \log_2(B_i)$,
\item {\bf Guessing Entropy:}
  $\GE(\Ss|\Ff) \rmdef (\frac{1}{2B}) \sum_{1 \leq i \leq k} B_i^2 + \frac{1}{2}$, and
  \item {\bf Min-Guess Entropy:}
  $\MG(\Ss|\Ff) \rmdef \min_{1 \leq i \leq k} \set{(B_i + 1)/2}$.
\end{enumerate}
\end{proposition}

\section{Shannon Mitigation Problem}
Our goal is to mitigate the information leakage due to the timing
side channels by adding synthetic delays to the program.
An aggressive, but commonly-used, mitigation strategy aims to eliminate the side
channels by adding delays such that every secret value yields a common
functional observation.
However, this strategy may often be impractical as it may result in unacceptable
performance degradations of the response time.
Assuming a well-known penalty function associated with the performance degradation,
we study the problem of maximizing entropy while respecting a bound on
the performance degradation.
We dub the decision version of this problem Shannon mitigation.

Adding synthetic delays to
execution-time of the program, so as to mask the side-channel,
can give rise to new functional observations that correspond to
upper-envelopes of various combinations of original observations.
Let $\Ff = \seq{f_1, f_2, \ldots, f_k}$ be the set of functional observations.
For $I \subseteq {1, 2, \ldots, k}$, let $f_I = \vy \in
\Real^m \mapsto \sup_{i\in I}  f_i(\vy)$ be the functional observation corresponding to upper-envelope of the functional observations in the set $I$.
Let $\Gg(\Ff) = \set{ f_I \::\: I \not = \emptyset \subseteq \set{1, 2, \ldots, k}}$ be the set of all possible functional observations resulting from
the upper-envelope calculations.
To change the observation of a secret value with functional observation $f_i$ to
a new observation $f_I$ (we assume that $i \in I$), we need to add delay
function $f^i_I: \vy \in \Real^m \mapsto f_I(y) - f_i(y)$.

\paragraph{Mitigation Policies.}
Let $\Gg \subseteq \Gg(\Ff)$ be a set of admissible post-mitigation
observations.
A {\it mitigation policy} is a function $\mu: \Ff \to \DIST(\Gg)$ that for each
secret $s \in \Ss_{f}$ suggests the probability distribution $\mu(f)$ over the
functional observations.
We say that a mitigation policy is {\it deterministic} if for all $f \in
\Ff$ we have that $\mu(f)$ is a point distribution.
Abusing notations, we represent a deterministic mitigation policy as a function
$\mu: \Ff \to \Gg$.
The semantics of a mitigation policy recommends to a program analyst
a probability $\mu(f)(g)$ to elevate a secret input $s \in \Ss_f$ from the
observational class $f$ to the class $g \in \Gg$ by adding
$\max \set{0, g(p) - f(p)}$ units delay to the corresponding execution-time
$\delta(s, p)$ for all $p \in Y$.
We assume that the mitigation policies respect the order, i.e. for  every mitigation
policy $\mu$ and for all $f \in \Ff$ and $g \in \Gg$, we have that $\mu(f)(g) >
0$ implies that $f \prec g$.
Let $M_{(\Ff\to\Gg)}$ be the set of mitigation policies from the set of
observational clusters $\Ff$ into the clusters $\Gg$.

For the functional observations $\Ff = \seq{f_1, \ldots, f_k}$ and a mitigation
policy $\mu \in M_{(\Ff\to\Gg)}$, the resulting observation set
$\Ff[\mu] \subseteq \Gg$ is defined as:
\[
\Ff[\mu] = \set{ g \in \Gg \::\: \text{ there exists } f \in \Ff \text{ such
    that } \mu(f)(g) > 0}.
\]
Since the mitigation policy is stochastic, we use average sizes of resulting
observations to represent fitness of a mitigation policy.
For $\Ff[\mu] = \seq{g_1, g_2, \ldots, g_\ell}$, we define their expected class
sizes $\Bb_\mu = \seq{C_1, C_2, \ldots, C_\ell}$ as
$C_i = \sum_{j=1}^{i} \mu(f_j)(f_i)\cdot B_j$ (observe that $\sum_{i=1}^{\ell} C_i = B$).
Assuming a uniform distribution on $\Ss$, various entropies for the expected
class size after applying a policy $\mu \in M_{(\Ff\to\Gg)}$ can be characterized
by the following expressions:
\begin{enumerate}
\item {\bf Shannon Entropy:}
  $\SE(\Ss|\Ff, \mu) \rmdef (\frac{1}{B})  \sum_{1 \leq i \leq \ell} C_i \log_2(C_i)$,
\item {\bf Guessing Entropy:}
  $\GE(\Ss|\Ff, \mu) \rmdef (\frac{1}{2B}) \sum_{1 \leq i \leq \ell} C_i^2 + \frac{1}{2}$, and
  \item {\bf Min-Guess Entropy:}
  $\MG(\Ss|\Ff, \mu) \rmdef \min_{1 \leq i \leq \ell} \set{(C_i + 1)/2}$.
\end{enumerate}
We note that the above definitions do not represent the expected entropies,
but rather entropies corresponding to the expected cluster sizes.
However, the three quantities provide bounds on the expected entropies after applying $\mu$.
Since Shannon and Min-Guess entropies are concave functions, from Jensen's
inequality, we get that  $\SE(\Ss|\Ff, \mu)$ and  $\MG(\Ss|\Ff, \mu)$ are upper
bounds on expected Shannon and Min-Guess entropies.
Similarly, $\GE(\Ss|\Ff, \mu)$, being a convex function, give a lower bound on
expected guessing entropy.

We are interested in maximizing the entropy while respecting constraints on the
overall performance of the system.
We formalize the notion of performance by introducing performance penalties:
there is a function $\pi: \Ff \times \Gg \to \Rplus$ such that elevating
from the observation $f \in \Ff$ to the functional observation $g \in
\Gg$ adds an extra $\pi(f, g)$ performance overheads to the program.
The expected performance penalty associated with a policy $\mu$, $\pi(\mu)$,
is defined as the probabilistically weighted sum of the penalties, i.e.
$\sum_{f \in \Ff, g \in \Gg:  f \prec g} |\Ss_{f}| \cdot
\mu(f)(g) \cdot \pi(f, g)$.
Now, we introduce our key decision problem.
\begin{definition}[Shannon Mitigation]
    \label{def-shannon-mitig}
    Given a set of functional observations $\Ff = \seq{f_1, \ldots, f_k}$, a set
    of admissible post-mitigation observations $\Gg \subseteq \Gg(\Ff)$,
    set of secrets $\Ss$, a penalty function $\pi: \Ff \times \Gg \to \Rplus$, a
    performance penalty upper bound $\Delta \in \Rplus$, and an entropy
    lower-bound $E \in \Rplus$, the Shannon mitigation problem
    $\textsc{Shan}_\Ee(\Ff, \Gg, \Ss, \pi, E, \Delta)$, for a given entropy measure
    $\Ee \in \set{\SE,\GE,\MG}$, is to decide whether there exists a
    mitigation policy $\mu \in M_{(\Ff\to\Gg)}$ such that $\Ee(\Ss| \Ff, \mu)
    \geq E$ and $\pi(\mu) \leq \Delta$.
    We define the {\it deterministic Shannon mitigation} variant where the goal
    is to find a deterministic such policy.
\end{definition}

%% file: mitigation.tex
\section{Algorithms for Shannon Mitigation Problem}
\label{sec:mit}
\subsection{Deterministic Shannon Mitigation}
We first establish the intractability of the deterministic variant.
\begin{theorem}
  Deterministic Shannon mitigation problem is NP-complete.
\end{theorem}

\begin{proof}
  It is easy to see that the deterministic Shannon mitigation problem is in NP:
  one can guess a certificate as a deterministic  mitigation policy $\mu \in M_{(\Ff\to\Gg)}$ and can
  verify in polynomial time that it satisfies the entropy and overhead
  constraints.
  Next, we sketch the hardness proof for the min-guess entropy measure by
  providing a reduction from the {\it two-way partitioning}
  problem~\cite{korf98acomplete}.
  For the Shannon entropy and guess entropy measures, a reduction can be
  established from the Shannon capacity problem~\cite{Fallg10} and the Euclidean
  sum-of-squares clustering problem~\cite{Aloi09}, respectively.

  Given a set $A = \set{a_1, a_2, \ldots, a_k}$ of integer values, the two-way
  partitioning problem is to decide whether there is a partition $A_1 \uplus A_2
  = A$ into two sets $A_1$ and $A_2$  with equal sums, i.e. $\sum_{a \in A_1} a =
  \sum_{a \in A_2} a$.
  W.l.o.g assume that $a_i \leq a_j$ for $i \leq j$.
  We reduce this problem to a deterministic Shannon mitigation problem
  $\textsc{Shan}_\MG(\Ff_A, \Gg_A, \Ss_A, \pi_A, E_A, \Delta_A)$  with $k$
  clusters $\Ff_A = \Gg_A = \seq{f_1, f_2, \ldots, f_k}$ with the secret set $\Ss_A =
  \seq{S_1, S_2, \ldots, S_k}$ such that $|S_i| = a_i$.
  If $\sum_{1 \leq i \leq k} a_i$ is odd then the solution to the two-way partitioning
  instance is trivially \texttt{no}.
  Otherwise, let $E_A = (1/2) \sum_{1 \leq i \leq k} a_i$.
  Notice that any deterministic mitigation strategy that achieves  min-guess entropy
  larger than or equal to $E_A$ must have at most two clusters.
  On the other hand, the best min-guess entropy value can be achieved by having
  just a single cluster.
  To avoid this and force getting two clusters corresponding to the two partitions
  of a solution to the two-way partitions problem instance $A$, we introduce performance penalties such that merging more than
  $k-2$ clusters is disallowed by keeping performance penalty $\pi_A(f, g) = 1$
  and performance overhead $\Delta_A = k-2$.
  It is straightforward to verify that an instance of the resulting min-guess entropy
  problem has a \texttt{yes} answer if and only if the two-way partitioning
  instance does.
\end{proof}

Since the deterministic Shannon mitigation problem is intractable,
we design an approximate solution for the problem.
Note that the problem is hard even if we only use existing functional
observations for mitigation, i.e., $\Gg = \Ff$. Therefore, we consider
this case for the approximate solution. Furthermore,
we assume the following {\it sequential dominance}
restriction on a deterministic policy $\mu$:
for $f, g \in \Ff$ if $f \prec g$ then either $\mu(f) \prec g$ or $\mu(f) = \mu(g)$.
In other words, for any given $f \prec g$,
$f$ can not be moved to a higher cluster than $g$ without having $g$
be moved to that cluster.
For example, Fig~\ref{Shannon-mitigation-examples}(a) shows
Shannon mitigation problem with four functional observations and all
possible mitigation policies (we represent $\mu(f_i)(f_j)$ with $\mu(i,j)$).
Fig~\ref{Shannon-mitigation-examples}(b)
satisfies the sequential dominance restriction,
while Fig~\ref{Shannon-mitigation-examples}(c) does not.

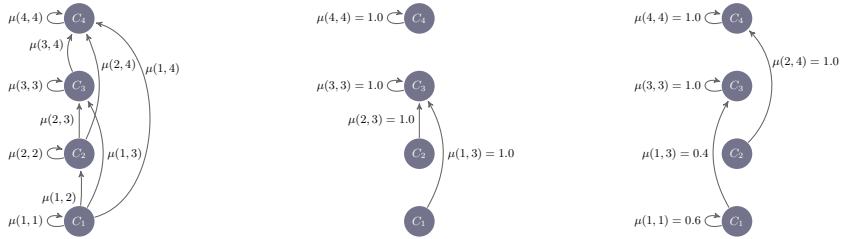
\begin{figure}[t!]
	\begin{minipage}{0.3\textwidth}
		\centering
		\scalebox{0.5}{
			\begin{tikzpicture}[->,>=stealth',shorten >=1pt,auto,node distance=1.8cm,thick,
			draw = black!60, fill=black!60]
			\tikzstyle{every state}=[fill=blue!10!gray,draw=none,text=white]

			\node[state] 		 (A)                              {$C_4$};
			\node[state]         (B) [below of=A]         {$C_3$};
			\node[state]         (C) [below of=B]         {$C_2$};
			\node[state]         (D) [below of=C]         {$C_1$};

			\path (D) edge  [loop left]  node {\small $\mu(1,1)$} (D)
			edge  [bend right=5]    node[pos=0.2, left] {\small $\mu(1,2)$} (C)
			edge  [bend right]    node[midway,right] {\small $\mu(1,3)$} (B)
			edge  [bend right=75]    node[pos=0.7, right] {\small $\mu(1,4)$} (A)
			(C) edge  [loop left]   node {\small $\mu(2,2)$} (C)
			edge         				  node[pos=0.5, left] {\small $\mu(2,3)$} (B)
			edge  [bend right=25]       node[pos=0.7, right] {\small $\mu(2,4)$} (A)
			(B) edge  [loop left]   node {\small $\mu(3,3)$} (B)
			edge  [bend left=25]  node[pos=0.7, left] {\small $\mu(3,4)$} (A)
			(A) edge [loop left]   node {\small $\mu(4,4)$} (A);
			\end{tikzpicture}
		}
	\end{minipage}
	\hfill
	\begin{minipage}{0.3\textwidth}
		\centering
		\scalebox{0.5}{
			\begin{tikzpicture}[->,>=stealth',shorten >=1pt,auto,node distance=1.8cm,thick,
			draw = black!60, fill=black!60]
			\tikzstyle{every state}=[fill=blue!10!gray,draw=none,text=white]

			\node[state] 		 (A)                              {$C_4$};
			\node[state]         (B) [below of=A]         {$C_3$};
			\node[state]         (C) [below of=B]         {$C_2$};
			\node[state]         (D) [below of=C]         {$C_1$};

			\path (D) edge  [bend right]  node[midway,right] {\small $\mu(1,3) = 1.0$} (B)
			(C) edge   node[pos=0.5, left] {\small $\mu(2,3) = 1.0$} (B)
			(B) edge  [loop left]   node {\small $\mu(3,3) = 1.0$} (B)
			(A) edge [loop left]   node {\small $\mu(4,4) = 1.0$} (A);
			\end{tikzpicture}
		}
	\end{minipage}
	\hfill
	\begin{minipage}{0.3\textwidth}
		\centering
		\scalebox{0.5}{
			\begin{tikzpicture}[->,>=stealth',shorten >=1pt,auto,node distance=1.8cm,thick,
			draw = black!60, fill=black!60]
			\tikzstyle{every state}=[fill=blue!10!gray,draw=none,text=white]

			\node[state] 		 (A)                              {$C_4$};
			\node[state]         (B) [below of=A]         {$C_3$};
			\node[state]         (C) [below of=B]         {$C_2$};
			\node[state]         (D) [below of=C]         {$C_1$};

			\path (D) edge  [loop left]  node {\small $\mu(1,1) = 0.6$} (D)
			edge  [bend left]    node[midway,left] {\small $\mu(1,3) = 0.4$} (B)
			(C) edge  [bend right=45]  node[pos=0.7, right] {\small $\mu(2,4) = 1.0$} (A)
			(B) edge  [loop left]   node {\small $\mu(3,3) = 1.0$} (B)
			(A) edge [loop left]   node {\small $\mu(4,4) = 1.0$} (A);
			\end{tikzpicture}
		}
	\end{minipage}
	\caption{(a). Example of Shannon mitigation problem with all possible mitigation policies
		for 4 classes of observations. (b,c) Two examples of the mitigation policies that results in 2
		and 3 classes of observations.}
	\label{Shannon-mitigation-examples}
\end{figure}

The search for the deterministic policies satisfying the sequential dominance
restriction can be performed efficiently using dynamic
programming by effective use of intermediate results' memorizations.

Algorithm~(\ref{alg:dyn-merge}) provides a pseudocode for the dynamic
programming solution to find a deterministic mitigation policy satisfying
the sequential dominance.
The key idea is to start with considering policies that produce a single
cluster for subclasses $P_i$ of the problem with the observation from
$\seq{f_1, \ldots, f_i}$, and then compute policies producing one additional cluster
in each step by utilizing the previously computed sub-problems and keeping track
of the performance penalties.
The algorithm terminates as soon as the solution of the current step respects
the performance bound.
The complexity of the algorithm is $O(k^3)$.
\input{algo}

\subsection{Stochastic Shannon Mitigation Algorithm}
Next, we solve the (stochastic) Shannon mitigation problem by posing it as an
optimization problem.
Consider the stochastic Shannon mitigation problem $\textsc{Shan}_\Ee$ $(\Ff, \Gg = \Ff,
\Ss_\Ff, \pi, E, \Delta)$ with a stochastic policy $\mu: \Ff \to \DIST(\Gg)$
and $\Ss_\Ff = \seq{S_1, S_2, \ldots, S_k}$.
The following program characterizes the optimization problem that solves the
Shannon mitigation problem with stochastic policy.

\begin{mdframed}
Maximize $\Ee$, subject to:
\begin{enumerate}
\item
  $0 \leq \mu(f_i)(f_j) \leq 1$ for $1 \leq i \leq j \leq k$
\item
  $\sum_{i \leq j \leq k} \mu(f_i)(f_j) = 1$ for all $1 \leq i \leq k$.
\item
  $\sum_{i = 1}^k \sum_{j =i}^k |S_i| \cdot \mu(f_i)(f_j)  \cdot \pi(f_i,f_j) \leq \Delta$.
\item
  $C_{j} = \sum_{i = 1}^j |S_i| \cdot \mu(f_i)(f_j)$ for $1 \leq j \leq k$.
\end{enumerate}
Here, the objective function $\Ee$ is one
of the following functions:
\begin{enumerate}
\item {\bf Guessing Entropy}
  $ \Ee_{GE} = \sum\limits_{j = 1}^k C_{j}^2$
\item {\bf Min-Guess Entropy}
  $\Ee_{MGE} = \min\limits_{1 \leq j \leq k} \{C_{j}~|~C_j > 0\}$
\item {\bf Shannon Entropy}
  $\Ee_{SE} = \sum\limits_{j = 1}^k C_{j} \cdot \log_{2}(C_{j})$
\end{enumerate}
\end{mdframed}

The linear constraints for the problem are defined as the following.
The condition (1) and (2) express that $\mu$ provides a  probability distributions,
condition (3) provides restrictions regarding the performance constraint, and
the condition (4) is the entropy specific constraint.
The objective function of the optimization problem is defined based on the
entropy criteria from $\Ee$. For the simplicity, we omit the constant terms from
the objective function definitions. For the guessing entropy, the problem is an
instance of linearly constrained quadratic optimization problem~\cite{nocedal2006numerical}.
The problem with Shannon entropy is a non-linear optimization problem~\cite{bertsekasnonlinear}.
Finally, the optimization problem with min-guess entropy is an instance of mixed integer
programming~\cite{nemhauser1988integer}. We evaluate the scalability of these
solvers empirically in Section~\ref{sec:micro} and leave the exact complexity as
an open problem. We show that the min-guess entropy objective function
can be efficiently solved with the branch and bound algorithms~\cite{papadimitriou1998combinatorial}.
Fig~\ref{Shannon-mitigation-examples}(b,c) show two instantiations of the
mitigation policies that are possible for the stochastic mitigation.

%% file: algo.tex
\begin{algorithm*}[t!]\normalsize
  {
    \DontPrintSemicolon
    \KwIn{The Shannon entropy problem $\textsc{Shan}_{MGE}(\Ff, \Gg = \Ff, \Ss_\Ff,
      \pi, E, \Delta)$}
    \KwOut{The entropy table ($T$).}
    \For{$i=1~to~k$}
	{
	  $T(i,1) = \Ee({\bigcup\limits_{j=1}^{i} S_{j}})$\;
	  \lIf{$\sum\limits_{1 \leq j \leq i}  \pi(j,i)(B_j/B) \leq \Delta$}{$\Pi(i,1) = \sum\limits_{1 \leq j \leq i}  \pi(j,i)(B_j/B)$}\lElse{$\Pi(i,1) = \infty$}
	}

	\lIf{$\Pi(k,1) < \infty$}
	    {
	      \Return $T$;
	    }

	    \For{${r}=2~to~k$}
		{
		  \For{$i=1~to~k$}
		      {
				$\Omega(i,{{r}})=\{j \::\: 1 \leq j < i \text{ and }
                      \Pi(j,{r}-1) +\sum\limits_{j < q \leq i} \pi(q,i)(B_q/B)  \leq \Delta\}$

                 \lIf{$\Omega{\neq}\emptyset$}
                     {
                       $T(i,{r}){=} \max\limits_{j \in \Omega(i,{r})}\Big(\min\big(T(j,{r}{-}1),
                       \Ee({\bigcup\limits_{q=j+1}^{i} S_{q}})\big)\Big)$}
                     \lElse{$T(i,{r}){=}-\infty$}

			       Let $j$ be the index that maximizes $T(i,{r})$

			       \lIf{$\Omega \neq
                                 \emptyset$}{$\Pi(i,{r}) = \big(\Pi(j,{r}-1) + \sum\limits_{j < q \leq i} \pi(q,i) (B_{q}/B) \big)$}
                                   \lElse{$\Pi(i,{r}) = \infty$}
		      }
		      \lIf{$\Pi(k,{r}) < \infty$}
			 {
			   \Return $T$;
			 }
	        }
        	\Return $T$;

	\caption{\textsc{Approximate Deterministic Shannon Mitigation}}
	\label{alg:dyn-merge}
}
\end{algorithm*}

%% file: experiments.tex
\section{Implementation Details}
\noindent \textbf{A. Environmental Setups.}
All timing measurements  are conducted on an Intel NUC5i5RYH.
We switch off JIT Compilation and run each experiment
multiple times and use the mean running time.
This helps to reduce the effects
of environmental factors such as the Garbage Collections.
All other analyses are conducted on an Intel i5-2.7 GHz machine.

\noindent \textbf{B. Implementation of Side Channel Discovery.}
We use the technique presented in~\cite{FuncSideChan18} for
the side channel discovery. The technique applies the functional data analysis ~\cite{ramsay2009functional} to create B-spline basis and
fit functions to the vector of timing observations for each secret value. Then,
the technique applies the functional data clustering~\cite{jacques2014functional} to
obtain $K$ classes of observations.
We use the number of secret values in a cluster as the class size metric
and the $L_1$ distance norm between the clusters as the penalty function.

\noindent \textbf{C. Implementation of Mitigation Policy Algorithms.}
\label{sec:impl-policy}
For the stochastic optimization, we
encode the Shannon entropy and guessing entropy with
linear constraints in Scipy~\cite{scipy}. Since the objective
functions are non-linear (for the Shannon entropy) and quadratic
(for the guessing entropy),
Scipy uses sequential least square programming
(SLSQP)~\cite{nocedal2006sequential} to maximize the
objectives. For the stochastic optimization with the min-guess
entropy, we encode the problem in Gurobi~\cite{gurobi}
as a mixed-integer programming (MIP) problem~\cite{nemhauser1988integer}.
Gurobi solves the problem efficiently with branch-and-bound
algorithms~\cite{MIP-branch-bound}. We use Java to implement
the dynamic programming.

\noindent \textbf{D. Implementation of Enforcement.}
The enforcement of mitigation policy is implemented
in two steps. \textit{First}, we use the initial timing functions
and characterize them with program internal properties
such as basic block calls. To do so, we use the decision
tree learning approach presented in~\cite{FuncSideChan18}.
The decision tree model characterizes each functional observations
with properties of program internals.
\textit{Second}, given the policy of mitigation,
we enforce the mitigation policy with a monitoring system implemented
on top of the Javassist~\cite{chiba1998javassist} library.
The monitoring system uses the decision tree model and matches the
properties enabled during an execution with the tree model (detection
of the current cluster). Then, it adds extra delays, based
on the mitigation policy, to the current execution-time and enforces
the mitigation policy. Note that the dynamic monitoring can result
in a few micro-second delays.
For the programs with timing differences in the order of micro-seconds,
we transform source code using the decision tree model.
The transformation requires manual efforts to modify and
compile the new program. But, it adds negligible delays.

\noindent \textbf{E. Micro-benchmark Results.}
\label{sec:micro}
Our goal is to compare different mitigation methods in terms of
their security and performance. We examine the
computation time of our tool \toolname in calculating the mitigation policies.
See appendix for the relationships between performance bounds
and entropy measures.

\noindent\textit{Applications}:
Mod\_Exp applications~\cite{mantel2015transforming} are instances
of square-and-multiply modular exponentiation ($R = y^k~mod~n$) used
for secret key operations in RSA~\cite{rivest1978method}.
Branch\_and\_Loop series consist of 6 applications where each application
has conditions over secret values and runs a linear loop over the public values.
The running time of the applications depend on the slope of the linear loops
determined by the secret input.

\noindent\textit{Computation time comparisons}:
Fig~\ref{fig:comp-time} shows the computation time
for Branch\_and \_Loop applications (the applications are
ordered in x-axis based on the discovered number of observational classes).
For the min-guess entropy, we observe that both stochastic and
dynamic programming approaches are efficient and fast as shown
in Fig~\ref{fig:comp-time}(a).
For the Shannon and guessing entropies,
the dynamic programming is scalable, while the stochastic mitigation
is computationally expensive beyond 60 classes of observations as
shown in Fig~\ref{fig:comp-time}(b,c).

\noindent\textit{Mitigation Algorithm Comparisons}:
Tab~\ref{tab:benchmark} shows micro-benchmark results
that compare the four mitigation algorithms with the two
program series. Double scheme mitigation
technique~\cite{askarov2010predictive}
does not provide guarantees on
the performance overhead, and we can see that it is increased
by more than 75 times for mod\_exp\_6. Double scheme
method reduces the number of classes of observations.
However, we observe that this mitigation has difficulty improving
the min-guess entropy.
Second, Bucketing algorithm~\cite{kopf2009provably} can
guarantee the performance overhead, but it is not an effective method
to improve the security of functional observations,
see the examples mod\_exp\_6 and Branch\_and\_Loop\_6.
Third, in the algorithms, \toolname guarantees the performance
to be below a certain bound, while it results in the highest entropy
values. In most cases, the stochastic optimization
technique achieves the highest min-entropy value. Here, we show
the results with min-guess entropy measure. Also,
we have strong evidences to show that \toolname achieves higher
Shannon and guessing entropies. For example, in
B\_L\_5, the initial Shannon entropy has improved from
$2.72$ to $6.62$, $4.1$, $7.56$, and $7.28$ for the double
scheme, the bucketing, the stochastic, and the deterministic
algorithms, respectively.

\begin{table*}[t!]
	\caption{Micro-benchmark results. M\_E and B\_L stand for Mod\_Exp
		and Branch\_and\_Loop applications. Legend:
		\#\textbf{S}: no. of secret values, \#\textbf{P}: no. of public values,
		$\Delta$: Upper bound over performance penalty,
		$\epsilon$: clustering parameter,
		\#\textbf{K}: classes of observations before mitigation,
		\#\textbf{K}$_X$: classes of observations after mitigation with X technique,
		\textbf{mGE}: Min-guess entropy before mitigation,
		\textbf{mGE}$_X$: Min-guess entropy after mitigation with X,
		\textbf{O}$_X$: Performance overhead added after mitigation with X.
	}
	\label{tab:benchmark}
		\resizebox{\textwidth}{!}{
		\begin{tabular}{ || l || r | r | r | r | r | r || r | r | r || r | r | r || r | r | r || r | r | r ||}
			\hline
			      &   \multicolumn{6}{c||}{Initial Characteristics} & \multicolumn{3}{c||}{Double Scheme}& \multicolumn{3}{c||}{Bucketing} & \multicolumn{3}{c||}{\toolname (Determ.)}& \multicolumn{3}{c||}{\toolname (Stoch.)} \\
			\cline{2-19}
			App(s) & \#\textbf{S} & \#\textbf{P} & $\Delta$ & $\epsilon~~$ & \#\textbf{K} & \textbf{mGE} & \#\textbf{K}$_{DS}$ & \textbf{mGE}$_{DS}$ & \textbf{O}$_{DS}$(\%) & \#\textbf{K}$_B$ & \textbf{mGE}$_{B}$ & \textbf{O}$_{B}$(\%) &  \textbf{K}$_{D}$ & \#\textbf{mGE}$_{D}$ & \textbf{O}$_{D}$(\%) & \#\textbf{K}$_{S}$ & \textbf{mGE}$_{S}$ & \textbf{O}$_{S}$(\%) \\ \hline
			M\_E\_1 & 32 & 32 & 0.5 & 1.0 & 1 & 16.5 & 1 & 16.5 & 0.0  & 1 & 16.5 & 0.0 & 1 & 16.5 & 0.0 &  1 & 16.5 & 0.0  \\ \hline
			M\_E\_2 & 64 & 64 & 0.5 & 1.0 & 2 & 16.5  & 1 & 32.5 & 5,221 & 1 & 32.5 & 27.6 & 1 & 32.5 & 21.4 &  1 & 32.5 & 21.4  \\ \hline
			M\_E\_3 & 128 & 128 & 0.5 & 2.0 & 2 & 32.5  & 1 & 64.5 & 5,407 & 1 & 64.5 & 33.9 & 1 & 64.5 & 22.7 &  1 & 64.5 & 22.7  \\ \hline
			M\_E\_4 & 256 & 256 & 0.5 & 2.0 & 4 & 10.5  & 1 & 128.5 & 6,679  & 1 &  128.5 & 30.7 & 1 & 128.5 & 28.3 & 1 & 128.5 & 28.3 \\ \hline
			M\_E\_5 & 512 & 512 & 0.5 & 5.0 & 23 & 1.0 & 1 & 256.5 & 7,294 & 2 & 128.5 & 50.0 & 1 & 256.5 & 31.0 & 1 & 253.0 & 30.3  \\ \hline
			M\_E\_6 & 1,024 & 1,024 & 0.5 & 8.0 & 40 & 1.0 & 1 & 512.5 & 7,822  & 20 & 1.0 & 34.5 & 2 & 27.5 & 46.7 & 5 & 85.5 & 50.0  \\ \hline
			B\_L\_1 & 25 & 50 & 0.5 & 10.0 & 4  & 3.0 & 3 & 3.0 & 73.0 & 3 & 3.0 & 17.5 & 2 & 5.5 & 26.1 &  2 & 6.5 & 34.9 \\ \hline
			B\_L\_2 & 50 & 50 & 0.5 & 10.0 & 8 & 3.0  & 4 & 3.0  & 61.3 & 5 & 3.0  &  21.9 &  2 & 10.5 & 45.3 & 2 & 13.0 & 45.3 \\ \hline
			B\_L\_3 & 100 & 50 & 0.5 & 20.0 & 16 & 3.0 & 4 & 8.0  & 42.4 & 8 & 3.0 &  33.4 &  2 & 20.5 & 48.3 & 2 & 21.5 & 50 \\ \hline
			B\_L\_4 & 200 & 50 & 0.5 & 20.0 & 32 & 3.0 & 6 & 3.0 & 36.9  & 16 & 3.0 &  28.7 &  2 & 48.0 & 48.7 & 2 & 50.5 & 49.7 \\ \hline
			B\_L\_5 & 400 & 50 & 0.5 & 20.0 & 64 & 3.0 & 8 & 3.0 & 35.4 & 32 & 3.0 & 27.2 &  3 & 65.5 & 32.0 & 2 & 100.5 & 50.0  \\ \hline
			B\_L\_6 & 800 & 50 & 0.5 & 20.0 & 125 & 3.0  & 12 & 8.0 & 37.8 & 29 & 3.0 &  52.5 &  3 & 133.0 & 34.6 &  2  & 200.5  & 49.6   \\ \hline
		\end{tabular}
	}
\end{table*}

\begin{figure}[t!]
	\centering
	\includegraphics[width=0.32\textwidth]{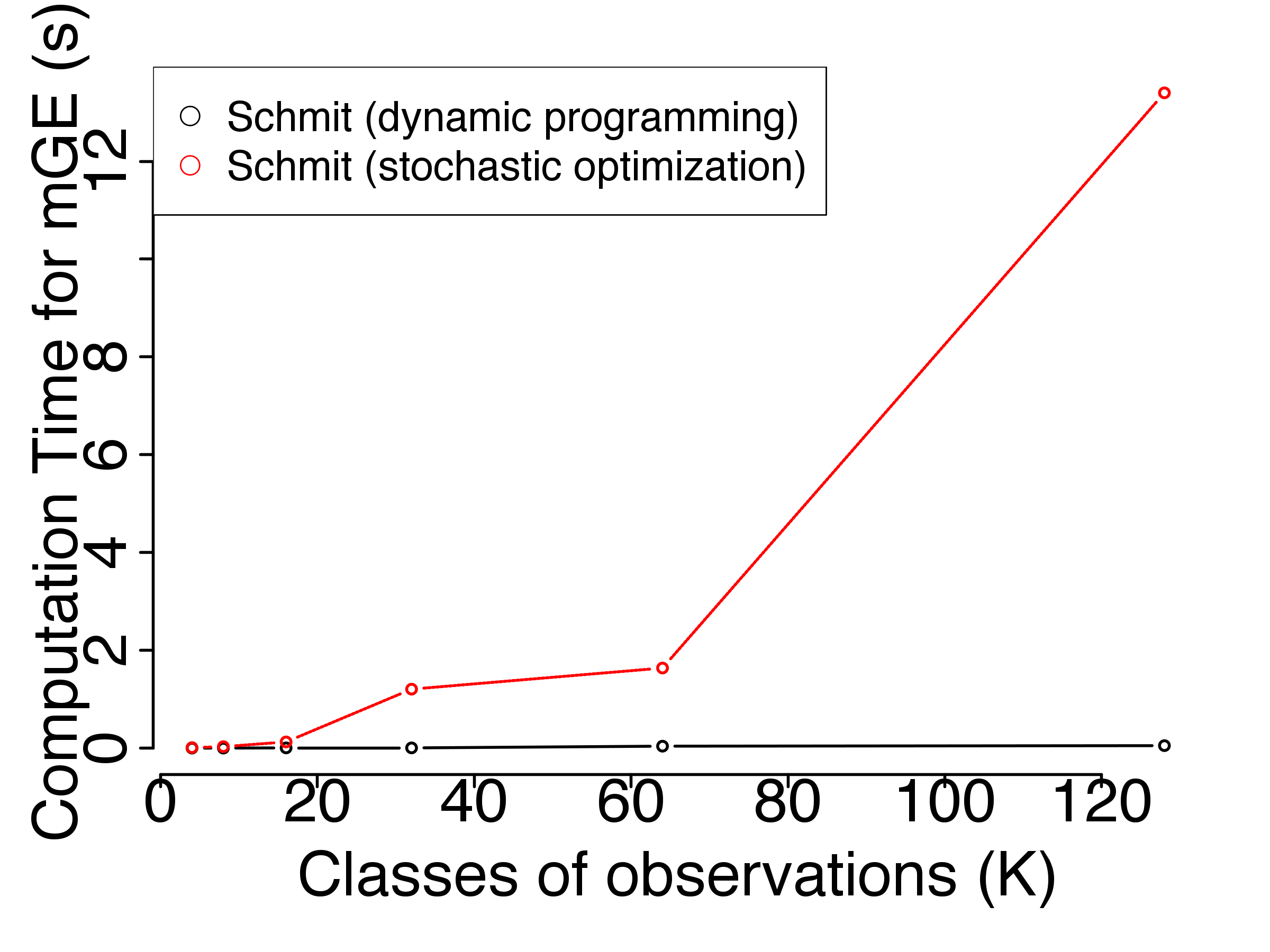}
	\includegraphics[width=0.32\textwidth]{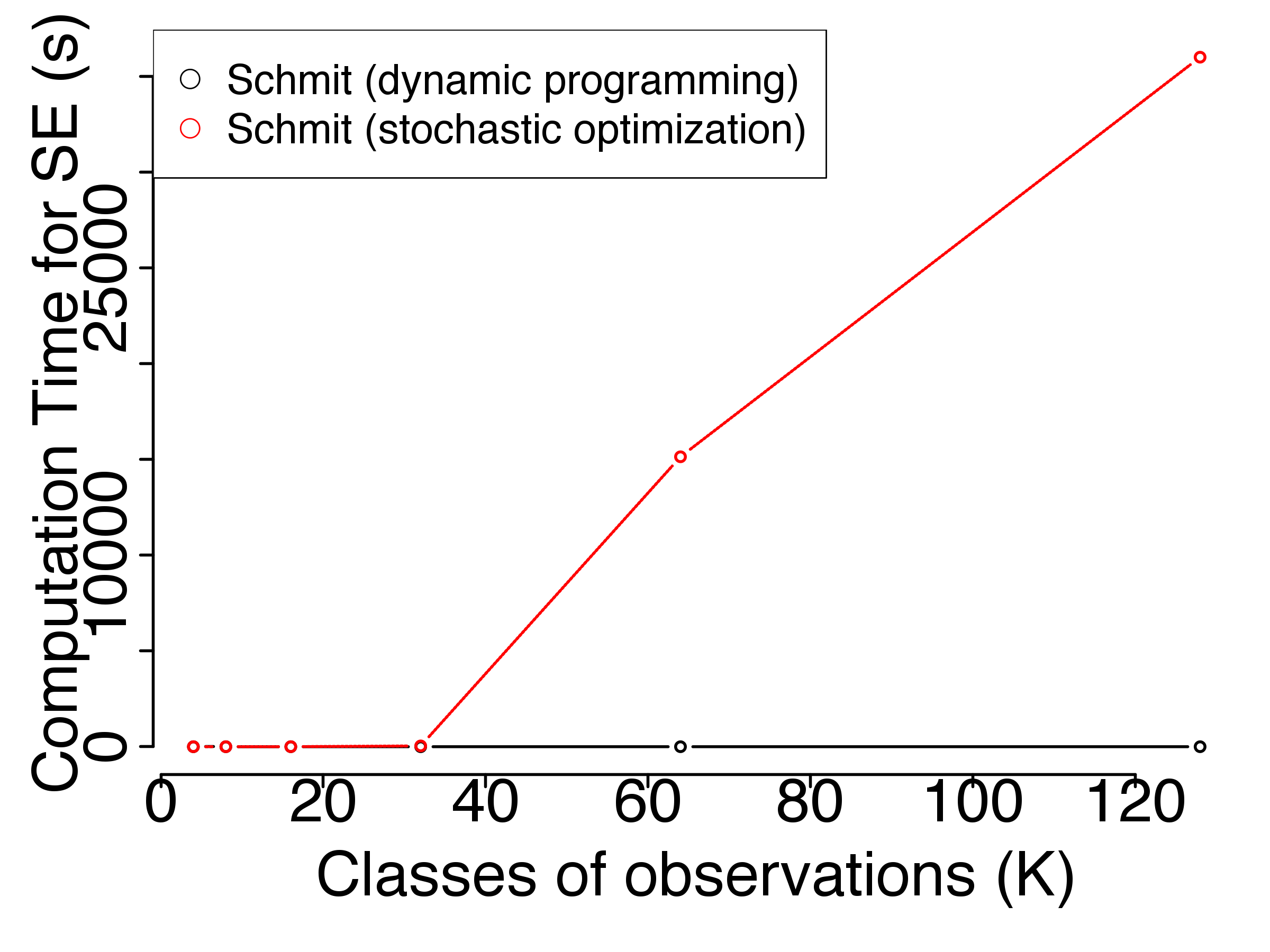}
	\includegraphics[width=0.32\textwidth]{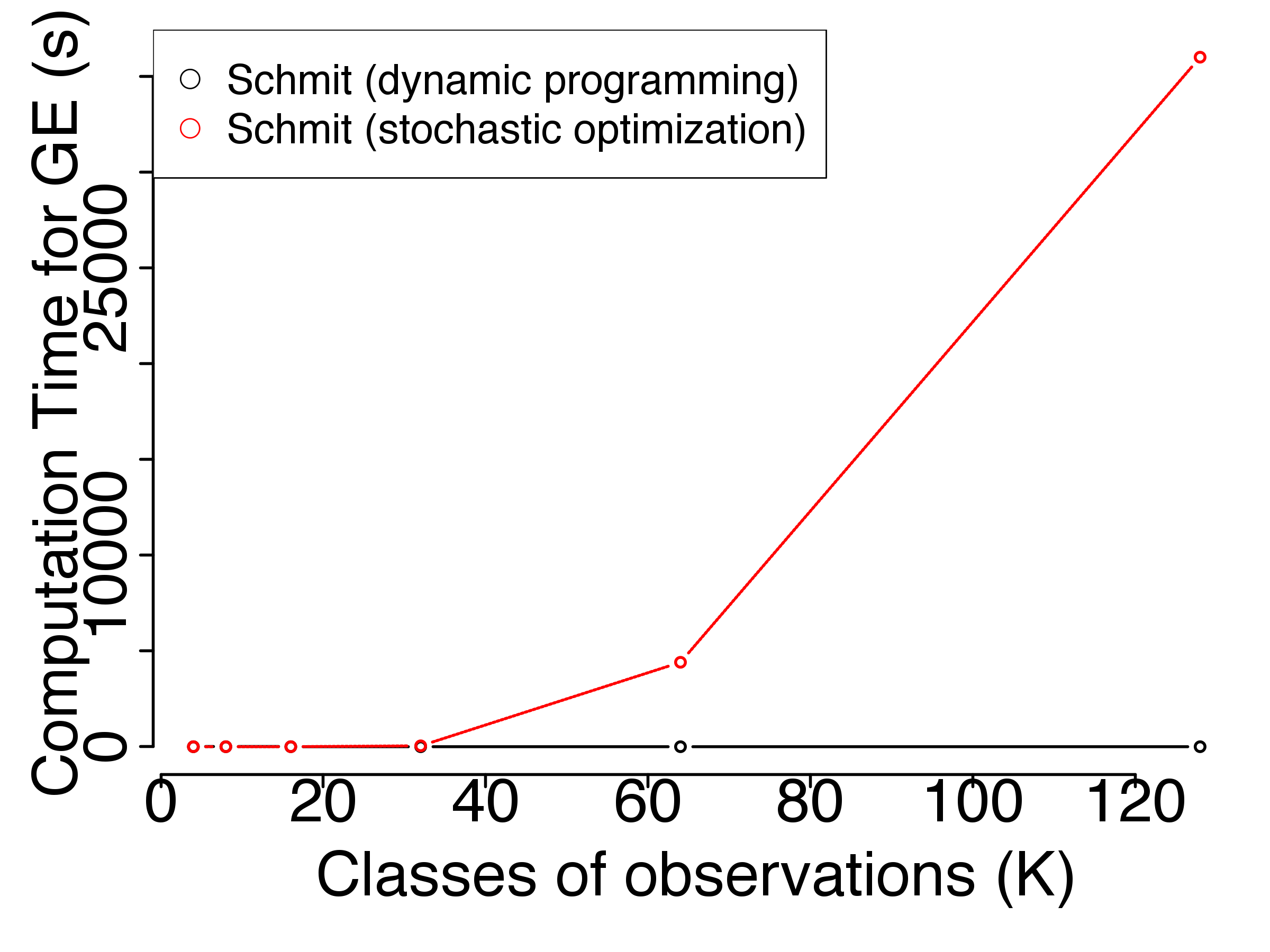}
	\caption{
		Computation time for synthesizing mitigation policy over
		Branch\_and\_Loop applications.
		Computation time for min-guess entropy (a) takes only few
		seconds. Computation time for the Shannon entropy (b)
		and guessing entropy (c) are expensive using Stochastic optimization.
		We set time-out to be 10 hours.
	}
	\label{fig:comp-time}
\end{figure}

%% file: case-study.tex
\section{Case Study}
\label{sec:case}
\noindent\noindent\textbf{Research Question.}
Does \toolname scale well and improve the security of applications
(entropy measures) within the given performance bounds?

\noindent\textbf{Methodology.}
We use the deterministic and stochastic algorithms
for mitigating the leaks.
We show our results for the min-guess entropy, but other
entropy measures can be applied as well.
Since the task is to mitigate existing leakages,
we assume that the secret and public inputs
are given.

\noindent\textbf{Objects of Study.}
We consider four real-world applications:

\vspace{-0.05em}
\hskip-0.6cm\textbf{}
\begin{tabular}{lrrrrrr} \toprule
	                                  & Num.       & Num.  &  Num.  & $\epsilon$~~~~  & Initial.  & Initial.~~   \\
	Application                &  Methods & Secret & Public &                            & clusters & Min-guess
	\\ \midrule
	GabFeed                         & 573~~ & 1,105 & 65~~~~ & 6.50~~~  & 34~~~  & 1.0~~~
	\\
	Jetty                         		& 63~~~ & 800~~ & 635~~~ & 0.1~~~~  & 20~~~  & 4.5~~~
	\\
	Java Verbal Expressions & 61~~~ & 2,000 & 10~~~~ & 0.02~~~ & 9~~~~ &  50.5~~
	\\
	Password Checker   		  &  6~~~~ &  20~~~ & 2,620~~ & 0.05~~~ & 6~~~~   & 1.0~~~
	\\ \bottomrule
\end{tabular}
\vspace{-0.05em}

\noindent In the inset table, we show the basic characteristics of these benchmarks.

\noindent\textit{GabFeed} is a chat server
with 573 methods~\cite{gabfeed}. There is a side channel
in the authentication part of the
application where the application takes users' public keys and its own private key,
and generating a common key~\cite{ChenFD17}. The vulnerability
leaks the number of set bits in the secret key.
Initial functional observations are shown in Fig~\ref{fig:gf-clusters}. There are 34
clusters and min-guess entropy is 1. We aim to maximize the min-guess
entropy under the performance overhead of 50\%.

\noindent\textit{Jetty.} We mitigate the side channels in
\texttt{util.security} package of Eclipse Jetty web server. The package
has \texttt{Credential} class which had a timing side channel.
This vulnerability was analyzed in \cite{ChenFD17} and fixed initially
in~\cite{jetty-1}. Then, the developers noticed that the implementation
in~\cite{jetty-1} can still leak information and fixed this issue with a
new implementation in~\cite{jetty-2}. However, this new implementation
is still leaking information~\cite{FuncSideChan18}. We  apply
\toolname to mitigate this timing side channels.
Initial functional observations is shown in Fig~\ref{fig:jetty-clusters}.
There are 20 classes of observations and the initial min-guess entropy is 4.5.
We aim to maximize the min-guess entropy under the performance overhead of 50\%.

\noindent\textit{Java Verbal Expressions} is
a library with 61 methods that construct regular expressions~\cite{JavaVerbalExp}.
There is a timing side channel in the library similar to
password comparison vulnerability~\cite{Keyczar} if the library
has secret inputs. In this case, starting from the initial character of
a candidate expression, if the character matches with the regular expression,
it slightly takes more time to respond the request than otherwise. This
vulnerability can leak all the regular expressions.
We consider regular expressions to have a maximum size of 9.
There are 9 classes of observations and the initial min-guess entropy is 50.5. We aim to
maximize the min-guess entropy under the performance overhead of 50\%.

\noindent\textit{Password Checker.} We consider the password
matching example from loginBad program~\cite{antonopoulos2017decomposition}.
The password stored in the server is secret, and the user's guess
is a public input.
We consider 20 secret (lengths at most 6) and 2,620 public inputs.
There are 6 different clusters, and the initial min-guess entropy is 1.

\begin{figure}[t!]
\begin{tabular*}{\linewidth}{@{\extracolsep{\fill}}ccc}
\\
\subfloat[Classes of observations]{\label{fig:gf-clusters}
	\includegraphics[width=0.3\linewidth]{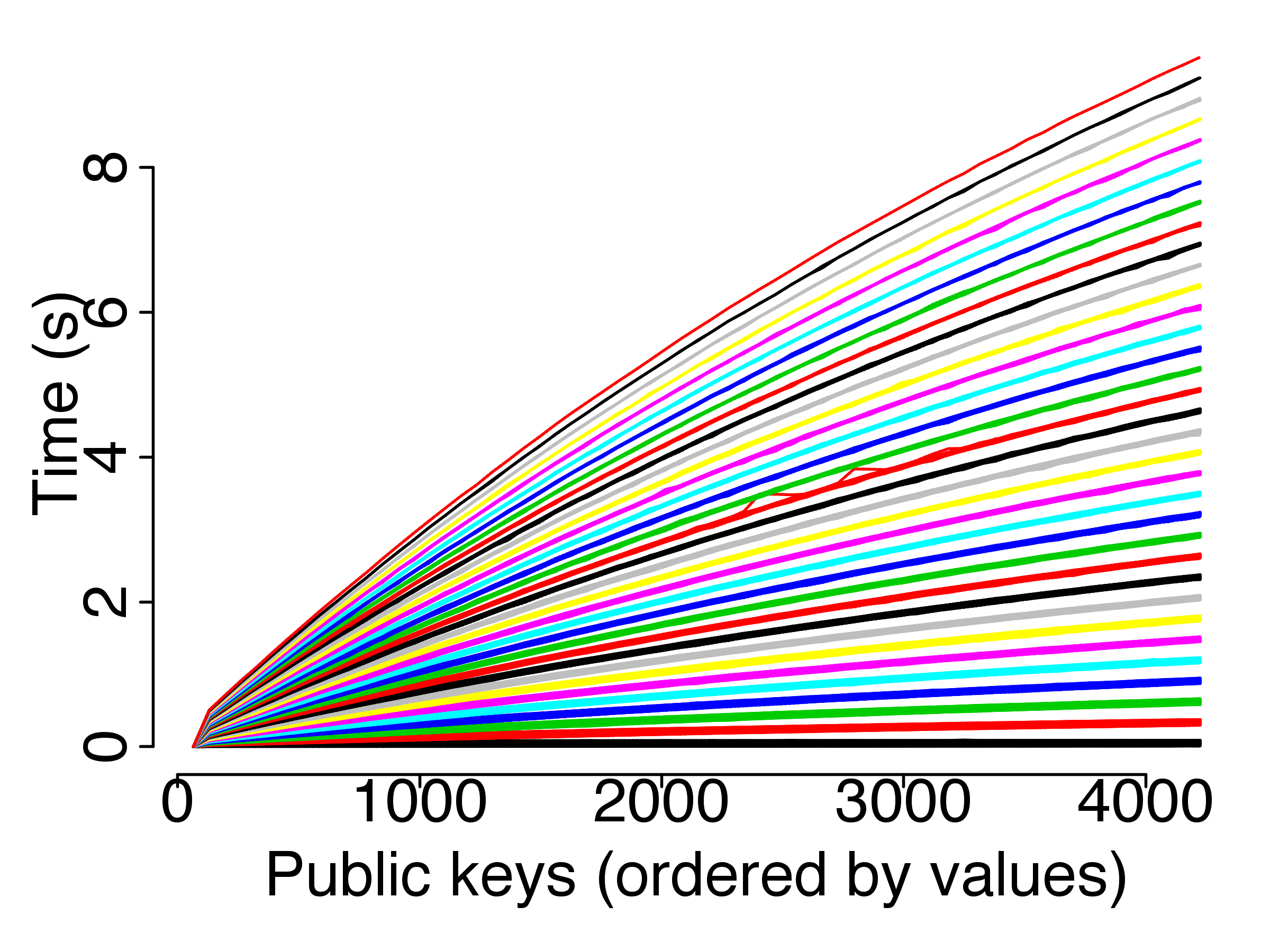}
}
&
\subfloat[Decision Tree]{\label{fig:gf-tree}
	\scalebox{0.5}{
	\vspace{-1em}
	\begin{tikzpicture}[align=center,node distance=0.8cm,->,thick,
	draw = black!60, fill=black!60]
	\centering
	\pgfsetarrowsend{latex}
	\pgfsetlinewidth{0.3ex}
	\pgfpathmoveto{\pgfpointorigin}

	\node[dtreenode,initial above,initial text={}] at (5,0) (l0)
	{OptimizedMultiplier.standard\\Multiply\_BasicBlock\_18};
	\node[dtreenode,below=of l0] (l2)
	{OptimizedMultiplier.standard\\Multiply\_BasicBlock\_18};
	\node[dtreenode,below=of l2] (l4)
	{OptimizedMultiplier.standard\\Multiply\_BasicBlock\_18};
	\node[below=of l4] (l6) {};

	\node[dtreeleaf,bicolor={black and black and 0.99},below left=of l0] (l1) {};
	\node[dtreeleaf,bicolor={red and red and 0.99},below left=of l2]  (l3) {};
	\node[dtreeleaf,bicolor={green and green and 0.99},below left=of l4] (l5) {};

	\path[->]  (l0) edge  node [left,pos=0.4] {$ = 3*y ~~$} (l1);
	\path  (l0) edge  node [right, pos=0.4] {$~~ \neq 3*y  $} (l2);
	\path  (l2) edge  node [left] {$ = 127*y ~~$} (l3);
	\path  (l2) edge  node [right] {$~~ \neq 127*y $} (l4);
	\path  (l4) edge  node [left] {$ = 251*y ~~$} (l5);
	\path  (l4) edge[dotted]  node [right] {$~~ \neq 251*y$} (l6);
	\end{tikzpicture}
	}
}
&
\subfloat[Mitigated Observations]{\label{fig:gf-mitigated}
	\includegraphics[width=0.30\linewidth]{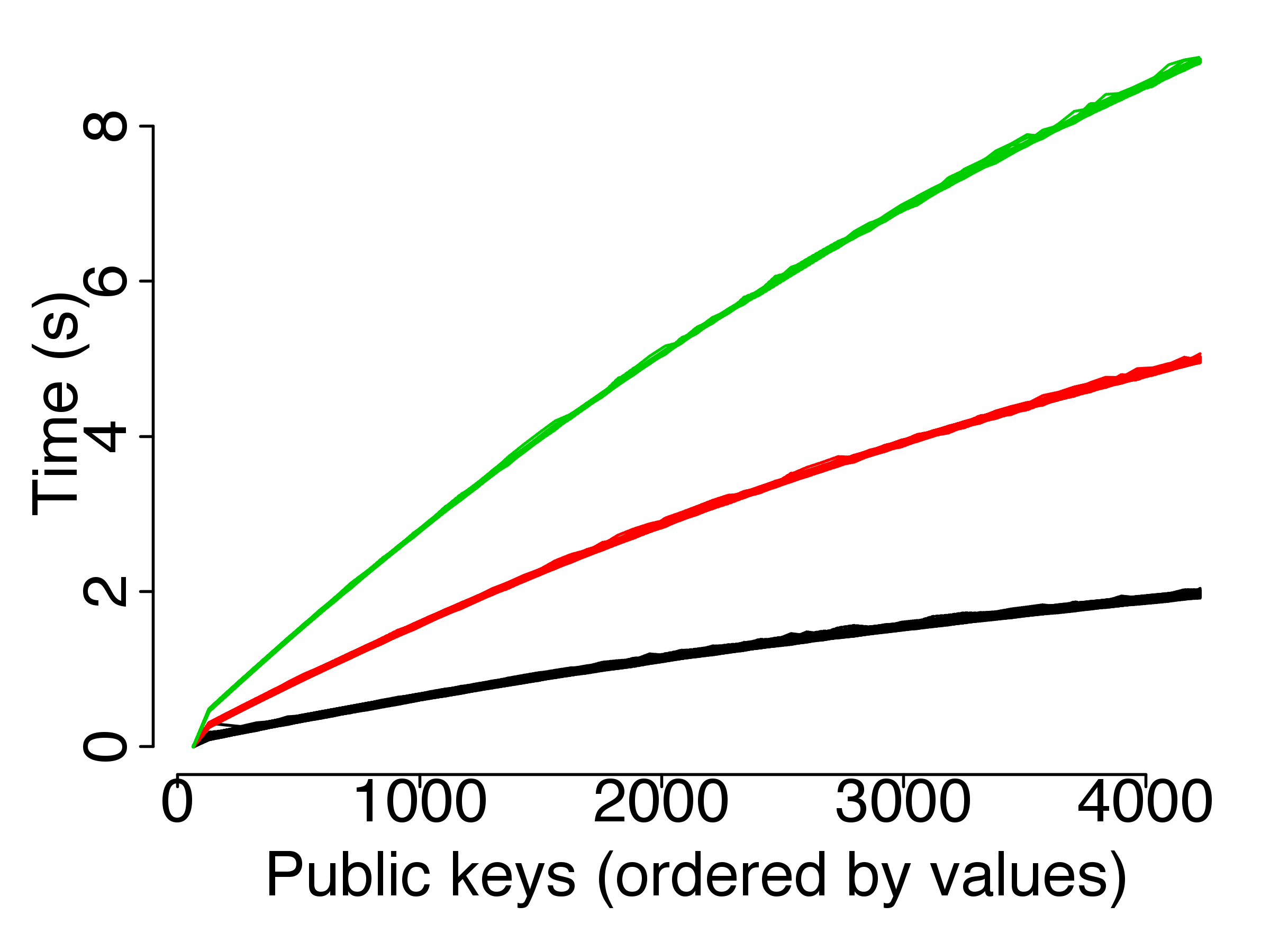}
}
\\
\subfloat[Classes of observations]{\label{fig:jetty-clusters}
	\includegraphics[width=0.3\linewidth]{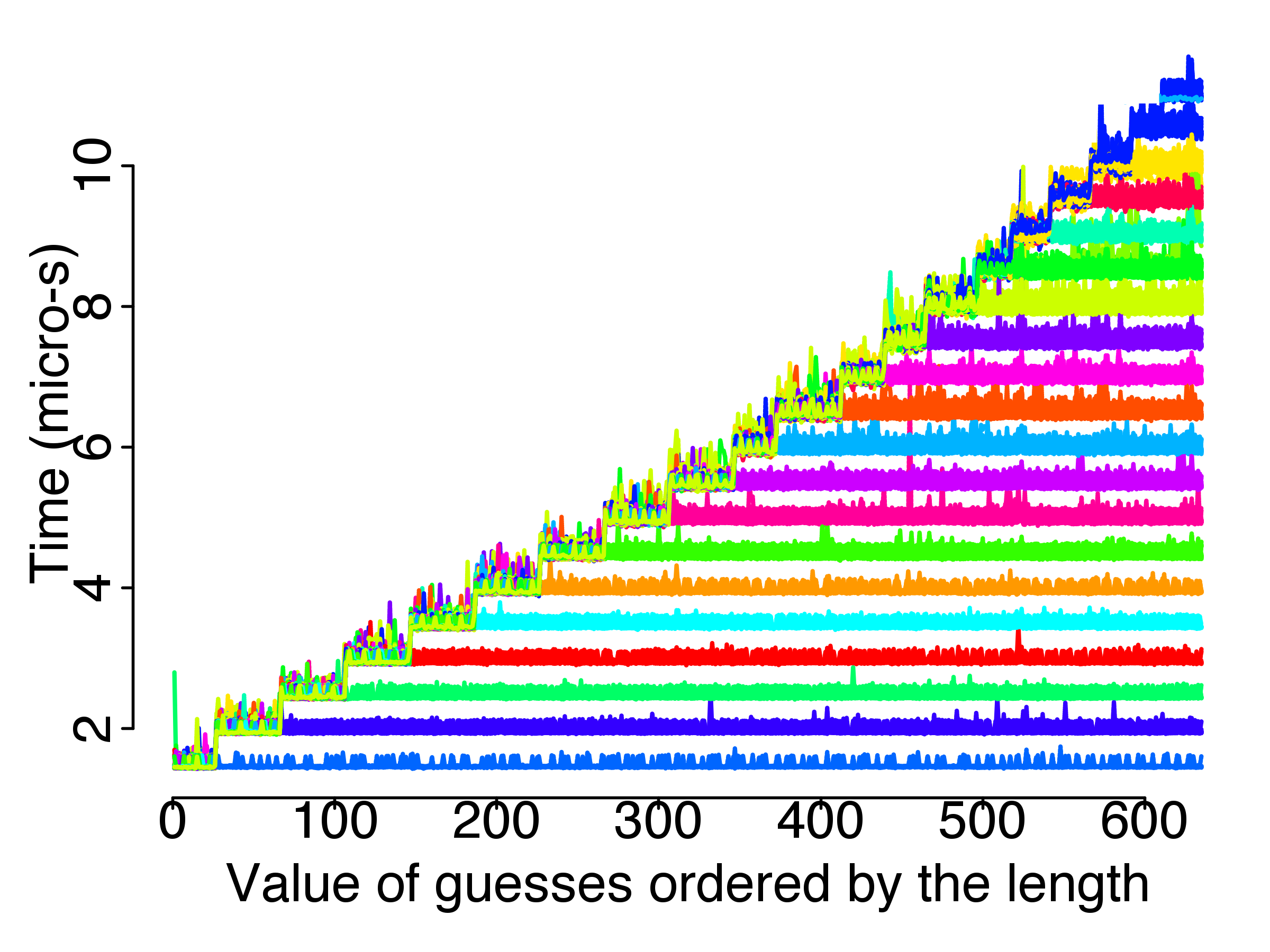}
}
&
\subfloat[Decision Tree]{\label{fig:jetty-tree}
	\scalebox{0.5}{
		\vspace{-1.0em}
    \begin{tikzpicture}[align=center,node distance=1cm,->,thick,
		draw = black!60, fill=black!60]
		\centering
		\pgfsetarrowsend{latex}
		\pgfsetlinewidth{0.3ex}
		\pgfpathmoveto{\pgfpointorigin}

		\node[dtreenode,initial above,initial text={}] at (0,0) (l0)  {
			jetty.util.security.\\Credential.stringEquals\_bblock\_106};
		\node[dtreenode,below=of l0] (l2)
		{jetty.util.security.\\Credential.stringEquals\_bblock\_106};
		\node[dtreenode,below=of l2] (l4)
		{jetty.util.security.\\Credential.stringEquals\_bblock\_106};
		\node[below=of l4] (l6) {};

		\node[dtreeleaf,bicolor={blue!60 and blue!60 and 0.99},below left=of
		l0] (l1) {};
		\node[dtreeleaf,bicolor={blue and blue and 0.99},below left=of l2]
		(l3) {};
		\node[dtreeleaf,bicolor={green and green and 0.99},below left=of
		l4] (l5) {};

		\path[->]  (l0) edge  node [left,pos=0.4] {$ = min(1,y) ~~$} (l1);
		\path  (l0) edge  node [right, pos=0.4] {$~~ \neq min(1,y) $} (l2);
		\path  (l2) edge  node [left] {$ = min(2,y) ~~$} (l3);
		\path  (l2) edge  node [right] {$~~ \neq min(2,y) $} (l4);
		\path  (l4) edge  node [left] {$ = min(3,y) ~~$} (l5);
		\path  (l4) edge[dotted]  node [right] {$~~ \neq min(3,y)$} (l6);
		\end{tikzpicture}
	}
}
&
\subfloat[Mitigated Observations]{\label{fig:jetty-mitigated}
	\includegraphics[width=0.30\linewidth]{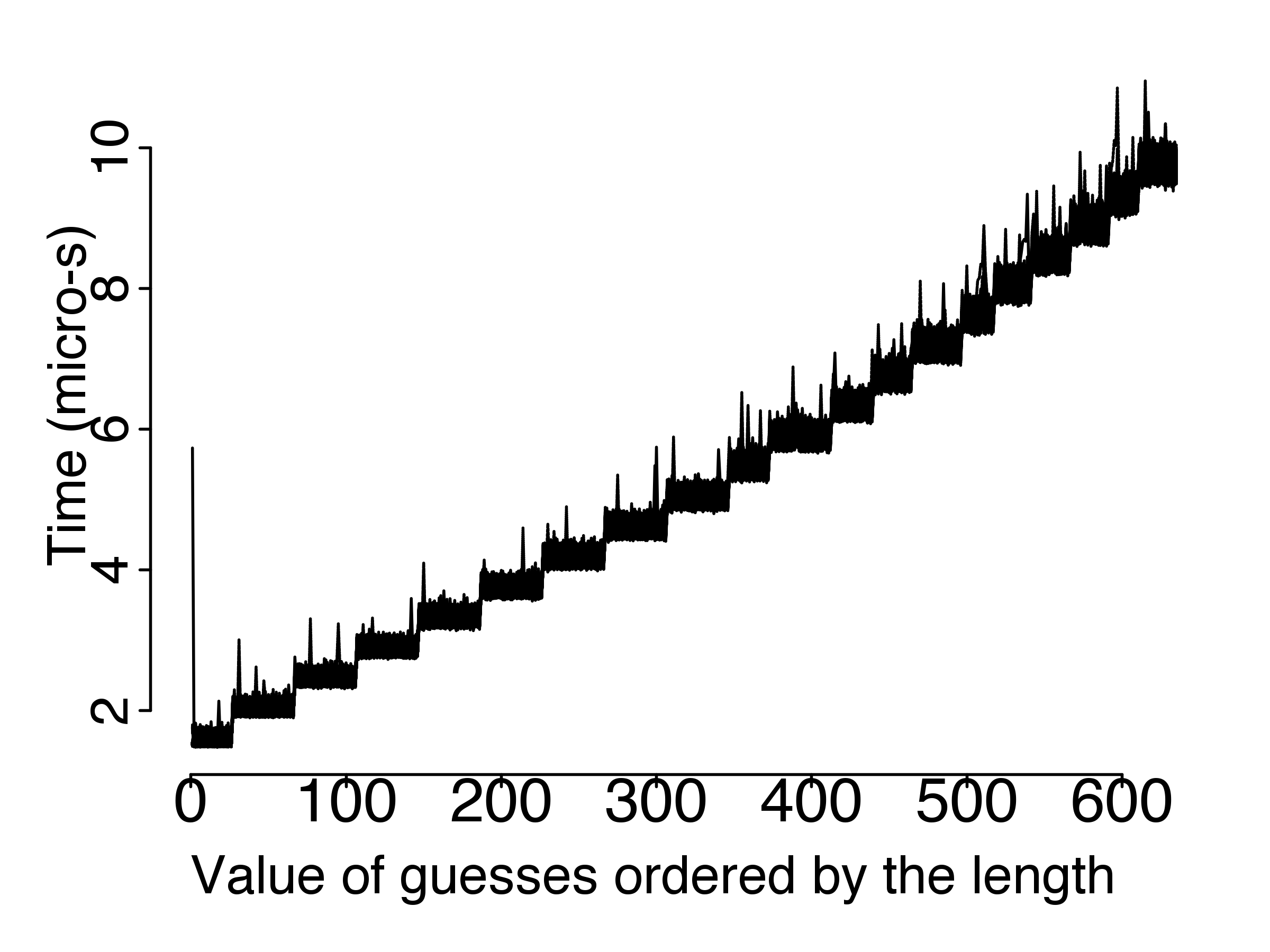}
}
\\
\subfloat[Classes of observations]{\label{fig:verbal-clusters}
	\includegraphics[width=0.3\linewidth]{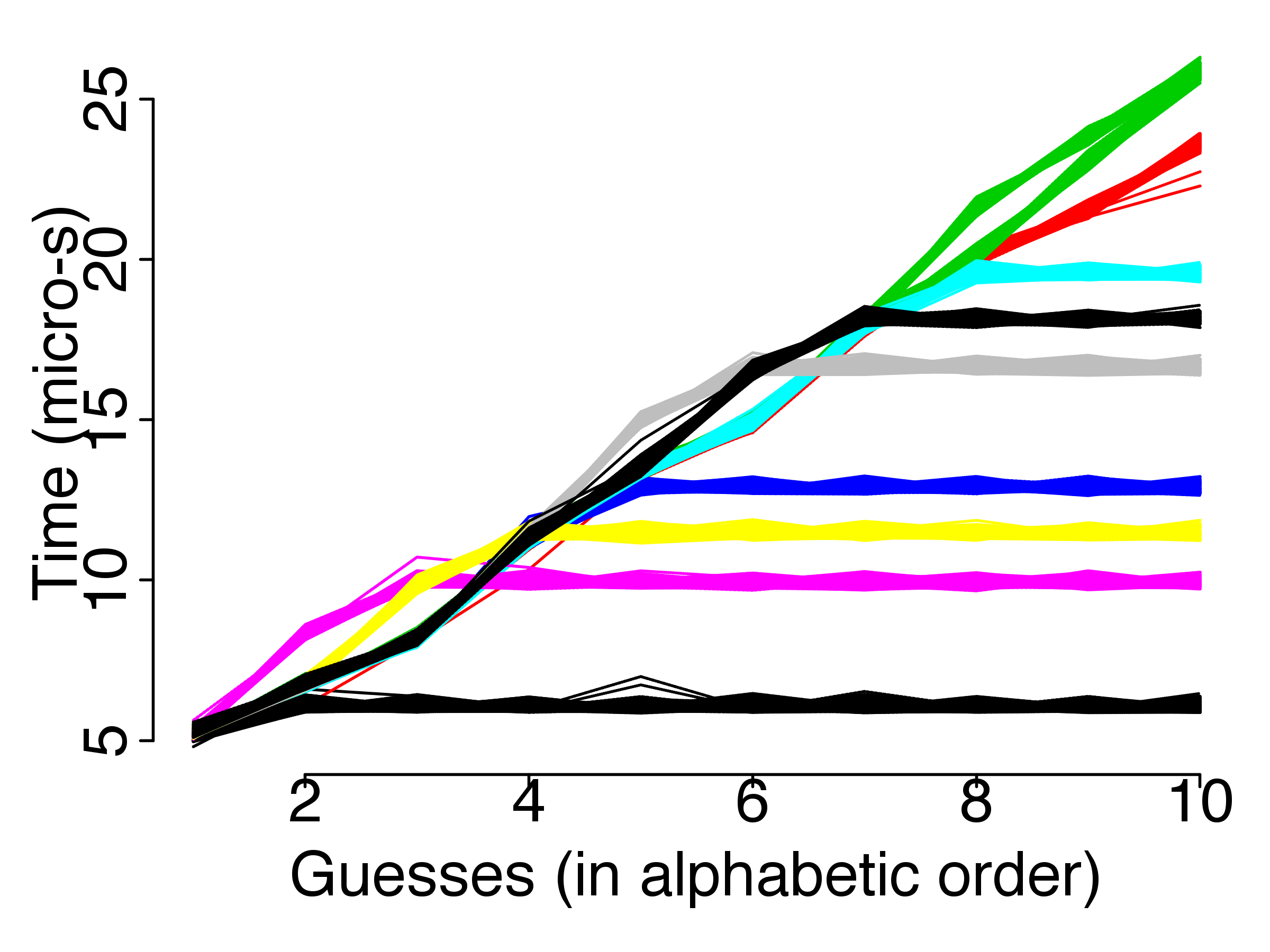}
}
&
\subfloat[Decision Tree]{\label{fig:verbal-tree}
	\scalebox{0.5}{
		\vspace{-1em}
		\begin{tikzpicture}[align=center,node distance=1cm,->,thick,
		draw = black!60, fill=black!60]
		\centering
		\pgfsetarrowsend{latex}
		\pgfsetlinewidth{0.3ex}
		\pgfpathmoveto{\pgfpointorigin}

		\node[dtreenode,initial above,initial text={}] at (0,0) (l0)  {
			verbalExp.example1\\.main.bblock\_1042};
		\node[dtreenode,below=of l0] (l2)
		{verbalExp.example1\\.main.bblock\_995};
		\node[dtreenode,below=of l2] (l4)
		{verbalExp.example1\\.main.bblock\_713};
		\node[below=of l4] (l6) {};

		\node[dtreeleaf,bicolor={black and black and 0.99},below left=of
		l0] (l1) {};
		\node[dtreeleaf,bicolor={gray and gray and 0.99},below left=of l2]
		(l3) {};
		\node[dtreeleaf,bicolor={cyan and cyan and 0.99},below left=of
		l4] (l5) {};

		\path[->]  (l0) edge  node [left,pos=0.4] {$ > 0 ~~$} (l1);
		\path  (l0) edge  node [right, pos=0.4] {$~~ = 0 $} (l2);
		\path  (l2) edge  node [left] {$ > 0 ~~$} (l3);
		\path  (l2) edge  node [right] {$~~ = 0 $} (l4);
		\path  (l4) edge  node [left] {$ > 0 ~~$} (l5);
		\path  (l4) edge[dotted]  node [right] {$~~ = 0$} (l6);
		\end{tikzpicture}
	}
}
&
\subfloat[Mitigated Observations]{\label{fig:verbal-mitigated}
	\includegraphics[width=0.30\linewidth]{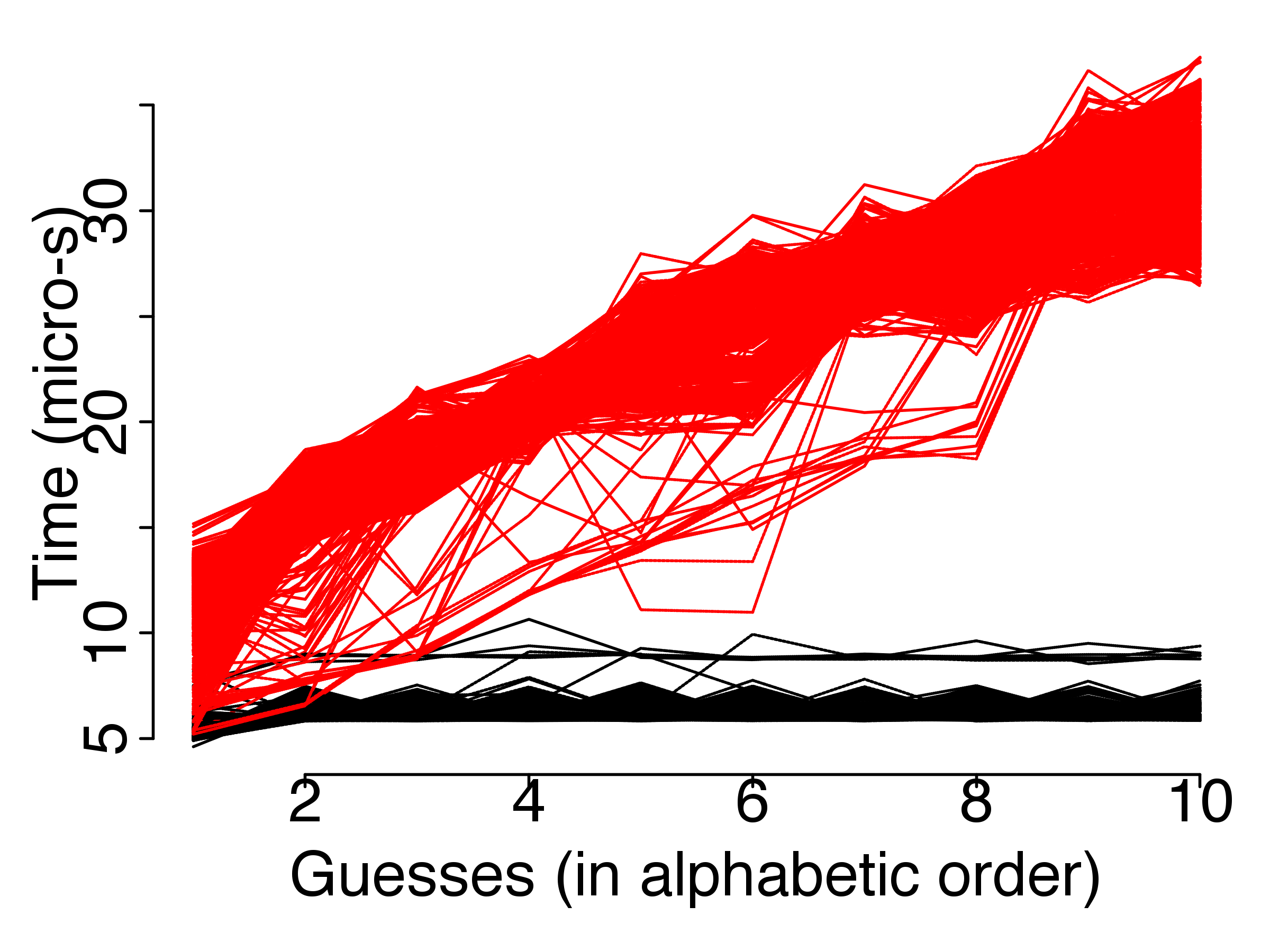}
}
\end{tabular*}
\caption{Initial functional observations, decision tree, and the mitigated observations
	from left to right for Gabfeed, Jetty, and Verbal Expressions
	from top to bottom.}
\label{fig:bigfig}
\end{figure}
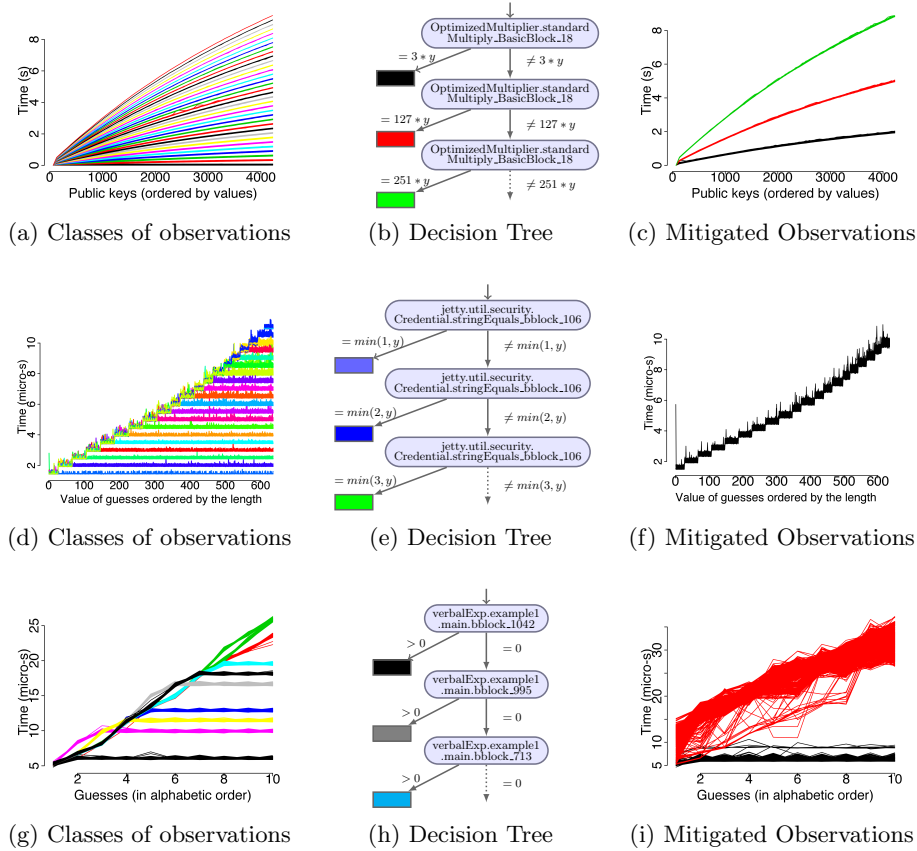

\noindent\textbf{Findings for GabFeed.}
With the stochastic algorithm, \toolname calculates the mitigation
policy that results in 4 clusters.
This policy improves the min-guess entropy from 1 to 138.5
and adds an overhead of 42.8\%.
With deterministic algorithm, \toolname returns 3 clusters.
The performance overhead is 49.7\% and the min-guess
entropy improves from 1 to 106.
The user chooses the deterministic policy and enforces the mitigation.
We apply CART decision tree learning and characterizes
the classes of observations with GabFeed method calls
as shown in Fig~\ref{fig:gf-tree}.
The monitoring system uses the decision tree model and automatically
detects the current class of observation. Then, it adds extra delays
based on the mitigation policy to enforce it. The results of
the mitigation is shown in Fig~\ref{fig:gf-mitigated}.
Answer for our research question. \textit{Scalability}:
It takes about 1 second to calculate the
stochastic and the deterministic policies.
\textit{Security}:
Stochastic and deterministic variants improve the
min-guess entropy more than 100
times under the given performance overhead of 50\%, respectively.

\noindent\textbf{Findings for Jetty.}
The stochastic algorithm and the deterministic algorithm
find the same policy that results in 1 cluster with 39.6\%
performance overhead. The min-guess entropy improves from 4.5 to 400.5.
For the enforcement, \toolname first
uses the initial clusterings and specifies their characteristics with program
internals that result in the decision tree model shown in Fig~\ref{fig:jetty-tree}.
Since the response time is in the order
of micro-seconds, we transform the source code using
the decision tree model by adding extra counter variables.
The results of the mitigation is shown in Fig~\ref{fig:jetty-mitigated}.
\textit{Scalability}:
It takes less than 1 second to calculate the policies for both algorithms.
\textit{Security}:
Stochastic and deterministic variants improve the min-guess entropy 89
times under the given performance overhead.

\noindent\textbf{Findings for Java Verbal Expressions.}
For the stochastic algorithm, the policy results in 2 clusters,
and the min-guess entropy has improved to 500.5.
The performance overhead is 36\%.
For the dynamic programming, the policy results in 2 clusters.
This adds 28\% of performance overhead, while
it improves the min-guess entropy from 50.5 to 450.5.
The user chooses to use the deterministic policy for the mitigation.
For the mitigation, we transform the source code using
the decision tree model and add the extra delays based on the mitigation
policy.

\noindent\textbf{Findings for Password Matching.}
Both the deterministic and the stochastic algorithms result in
finding a policy with 2 clusters where the min-guess entropy has
improved from 1 to 5.5 with the performance overhead of 19.6\%.
For the mitigation, we transform the source code using
the decision tree model and add extra delays based on the mitigation
policy if necessary.

%% file: related.tex
\section{Related Work}
\label{sec:related}
Quantitative theory of information have been widely used
to measure how much information is being leaked with side-channel
observations~\cite{smith2009foundations,KB07,backes2009automatic,heusser2010quantifying}.
Mitigation techniques increase the remaining entropy of secret sets
leaked through the side channels, while considering the performance
~\cite{kopf2009provably,askarov2010predictive,zhang2011predictive,zhang2012language,kadloor2012mitigating,schinzel2011efficient}.

K{\"o}pf and D{\"u}rmuth~\cite{kopf2009provably} use a
bucketing algorithm to partition programs' observations
into intervals. With the unknown-message threat model,
K{\"o}pf and D{\"u}rmuth~\cite{kopf2009provably} propose a dynamic programming
algorithm to find the optimal number of possible observations
under a performance penalty.
The works~\cite{askarov2010predictive,zhang2011predictive}
introduce different black-box schemes to mitigate leaks.
In particular, Askarov et al.~\cite{askarov2010predictive}
show the quantizing time techniques,
which permit events to release at scheduled constant slots,
have the worst case leakage if the slot is not filled with events.
Instead, they introduce the double scheme method that
has a schedule of predictions like the quantizing
approach, but if the event source fails to
deliver events at the predicted time, the
failure results in generating a new schedule
in which the interval between predictions is doubled.
We compare our mitigation technique with both algorithms throughout
this paper.

Elimination of timing side channels is a common technique
to guarantee the confidentiality of software~\cite{agat2000transforming,molnar2005program,wu2018eliminating,eldib2014synthesis,kopf2007transformational,mantel2015transforming}.
The work~\cite{wu2018eliminating} aims to eliminate side channels using
static analysis enhanced with various techniques to keep the performance
overheads low without guaranteeing the amounts of overhead. In contrast,
we use dynamic analysis and allow a small amount of information to leak,
but we guarantee an upper-bound on the performance overhead.

Machine learning techniques have been used for explaining
timing differences between traces
~\cite{song2014statistical,tizpaz2017discriminating,aaai18}.
Tizpaz-Niari et al.~\cite{aaai18} consider performance issues in
softwares. They also cluster execution times of programs and
then explain what program properties distinguish
the different functional clusters. We adopt their techniques
for our security problem.

%% file: ack.tex
\paragraph{Acknowledgements.}
The authors would like to thank Mayur Naik for shepherding our paper and
providing useful suggestions.
This research was supported by DARPA under agreement FA8750-15-2-0096.

%% file: appendix.tex
\section{Appendix}

\subsection{Overview of \toolname}
\label{exp:overview}

\toolname consists of three components:

\noindent\textbf{1) Initial Security Analysis.}

\noindent Inspired by~\cite{FuncSideChan18},
for each secret value, we use B-spline
basis~\cite{ramsay2006functional} in general to model
arbitrary timing functions of secret values in the domain
of public inputs, but we also allow simpler functional models
such as polynomial functions. We use the
non-parametric functional clustering~\cite{ferraty2006nonparametric}
with hierarchal algorithms~\cite{johnson1967hierarchical}
to obtain the initial classes of observations or clusters.
The clustering algorithm groups timing functions that are
$\epsilon$ close to each other in the same cluster.
The size of class is the number of secret values
in the cluster. The $l-$norm distance
between clusters forms the penalty matrix.

\noindent\textit{Highlight.}
This step finds the classes of observations over
secret values using functional clustering and returns
the label (cluster) of each secret value and
the distance (as a penalty) between the clusters.

\begin{figure}
	\centering
	\includegraphics[width=1.0\textwidth]{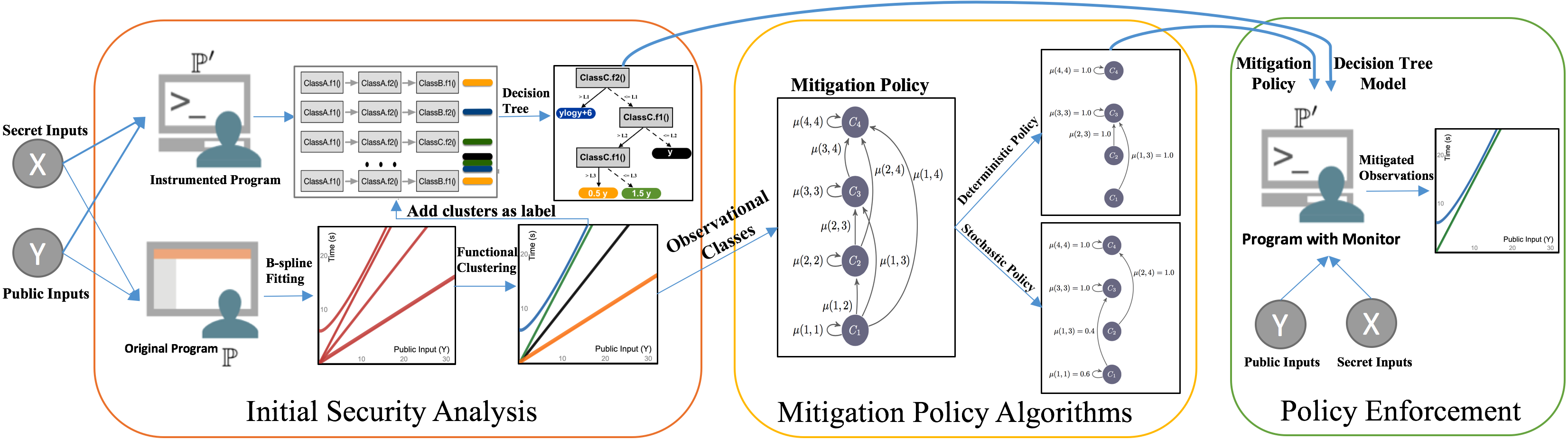}
	\caption{
		\toolname work-flow. \toolname consists of three components.
	}
	\label{fig:schmit-work-flow}
\end{figure}

\noindent\textbf{2) Mitigation policy.}

\noindent We uses the policy algorithms (Section~5)
to calculate the mitigation policy given the clusters,
their sizes, and their distances.
We use two types of algorithms: deterministic and
stochastic. The deterministic algorithm is an instance
of dynamic programming implemented in Java.
The stochastic algorithms have three variants for three
types of information theory measure. The variant
based on Min-guess entropy is the main emphasis
in this paper that implemented using Gurobi~\cite{gurobi}.
The two other variants (for Shannon and Guessing entropies) are implemented
in python using Scipy library~\cite{scipy}. See Section~6(C) for further details.

\noindent\textit{Highlight.}
This step calculates the mitigation policy that shows
how to merge different clusters to maximize an information
theory criterion given an upper-bound on the amount of
performance overhead.

\noindent\textbf{3) Enforcement of mitigation policy.}

\noindent In the first step, we characterize each class of observation with
program internal properties.
We use decision tree algorithms~\cite{Breiman/1984/CART}
to characterize each class of observation with corresponding
program internal features.
Fig~2(b) in Section~2 is an example of decision tree model
that characterizes each class of observation of Fig~1(b) in Section~2
with the basic block invocations at line 16 of \texttt{modExp} method.
In the second step, we enforce the mitigation policy.
This step can be done either with a monitoring system at run-time
automatically or with a source code transformation semi-automatically.
The enforcement uses the decision tree model and matches the
properties enabled during an execution with the tree model.
Then, it adds extra delays, based on the mitigation policy,
to the execution in order to enforce the mitigation policy.
The result of mitigation can be verified by applying the clustering
algorithm on the mitigated execution times.

\noindent\textit{Highlight.}
This step uses the functional clusters and the decision
tree model and enforces the mitigation policy either
with a monitoring system at run-time or souce code
transformations. The clustering algorithm over the
mitigated execution times can be used to verify the
mitigation model.

\subsection{Exponential blow-up for functional mitigations.}
Figure~\ref{fig:mitigation-2} shows possible observational
classes $C_S$ for $S \not= \emptyset \subseteq \set{1, 2, 3}$
for three observation classes $C_1, C_2$, and $C_3$.
The cluster $C_S$ corresponds to the observation class with
the execution-time equals to the upper-envelope of all functions
from the classes in $S$.

\begin{figure}[!t]
	\centering
  \includegraphics[width=0.48\textwidth]{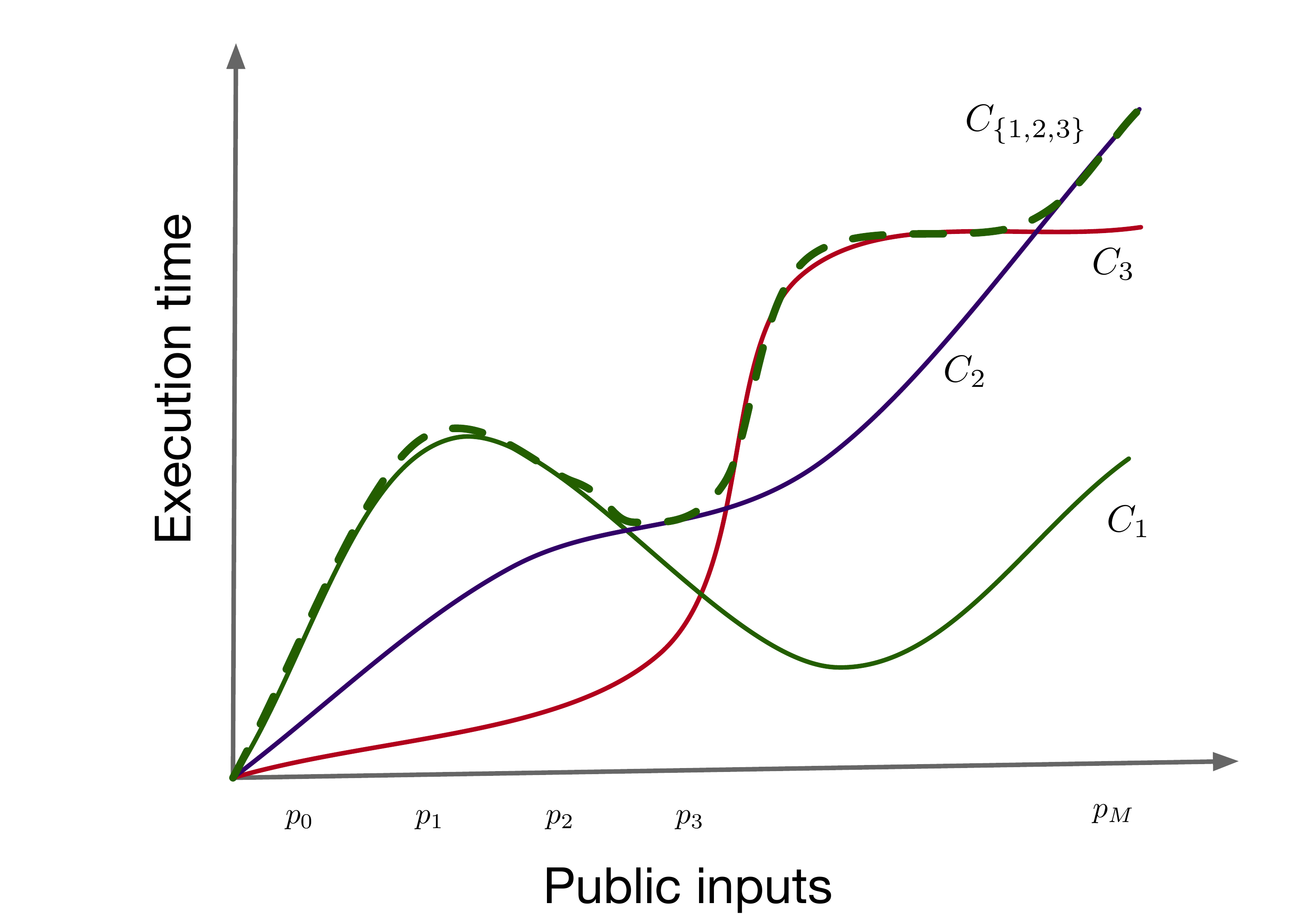}\label{fig:mitigation-1}
  \includegraphics[width=0.48\textwidth]{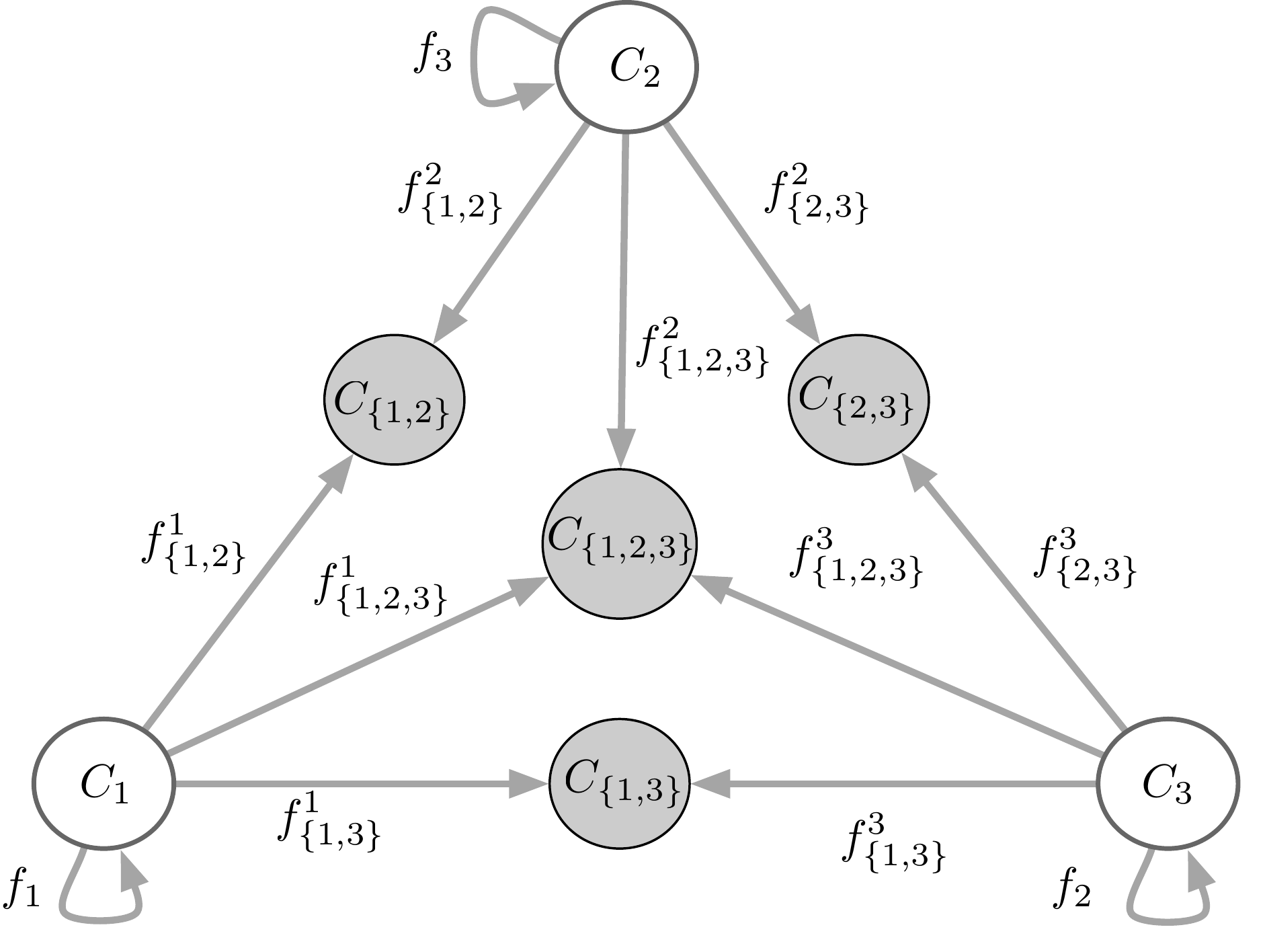}\label{fig:mitigation-2}
	\caption{
		(a) A program with three functional observations $C_1, C_2$ and
		$C_3$. The dashed curve $C_{1,2,3}$ is the upper-envelope of functions
		$C_1$, $C_2$, and $C_3$.
		(b) All possible resulting combinations of functional observations resulting
		from adding delays to various clusters.
	}
\end{figure}

\begin{figure}[t!]
	\centering
	\includegraphics[width=0.32\textwidth]{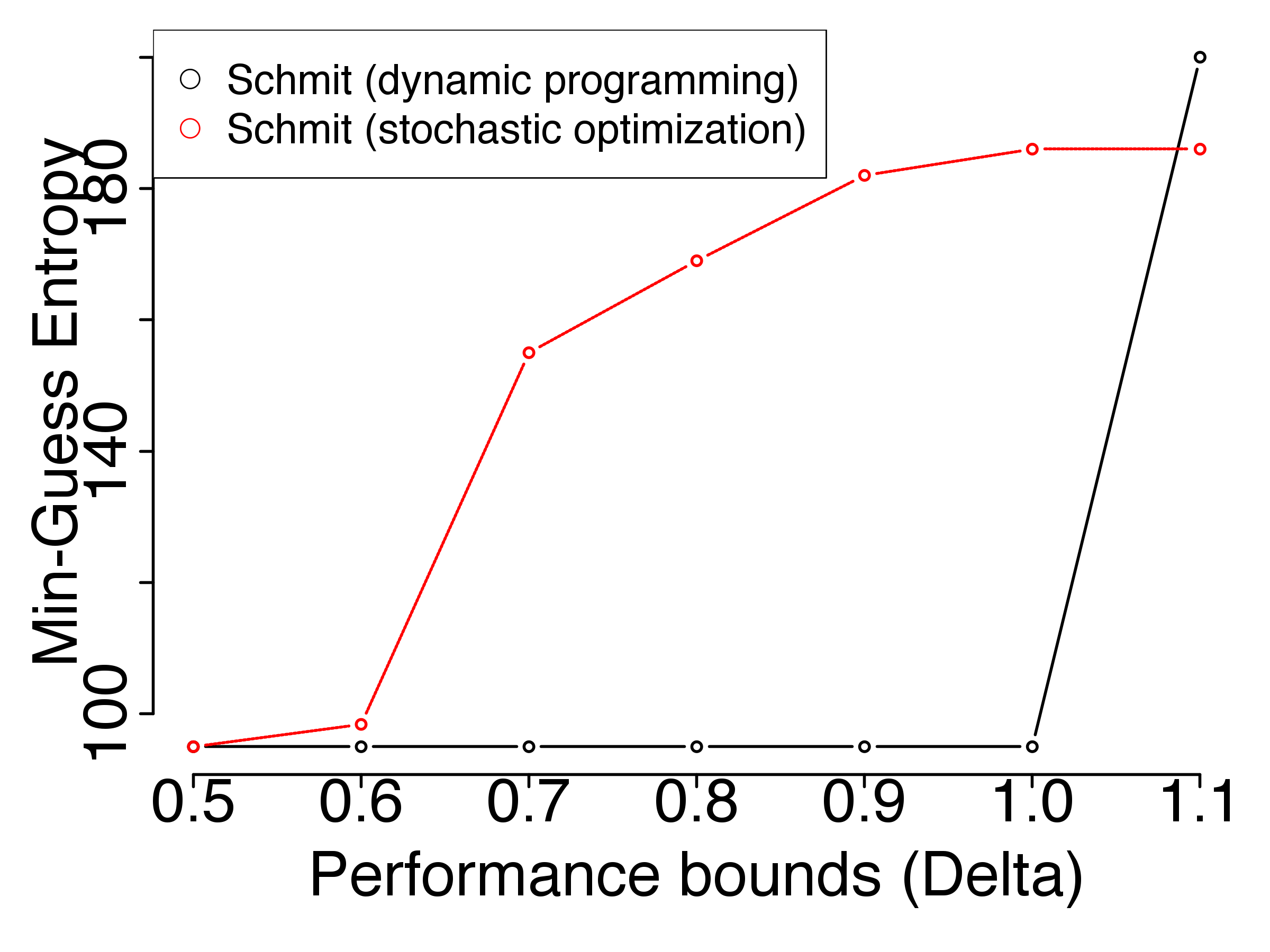}
	\includegraphics[width=0.32\textwidth]{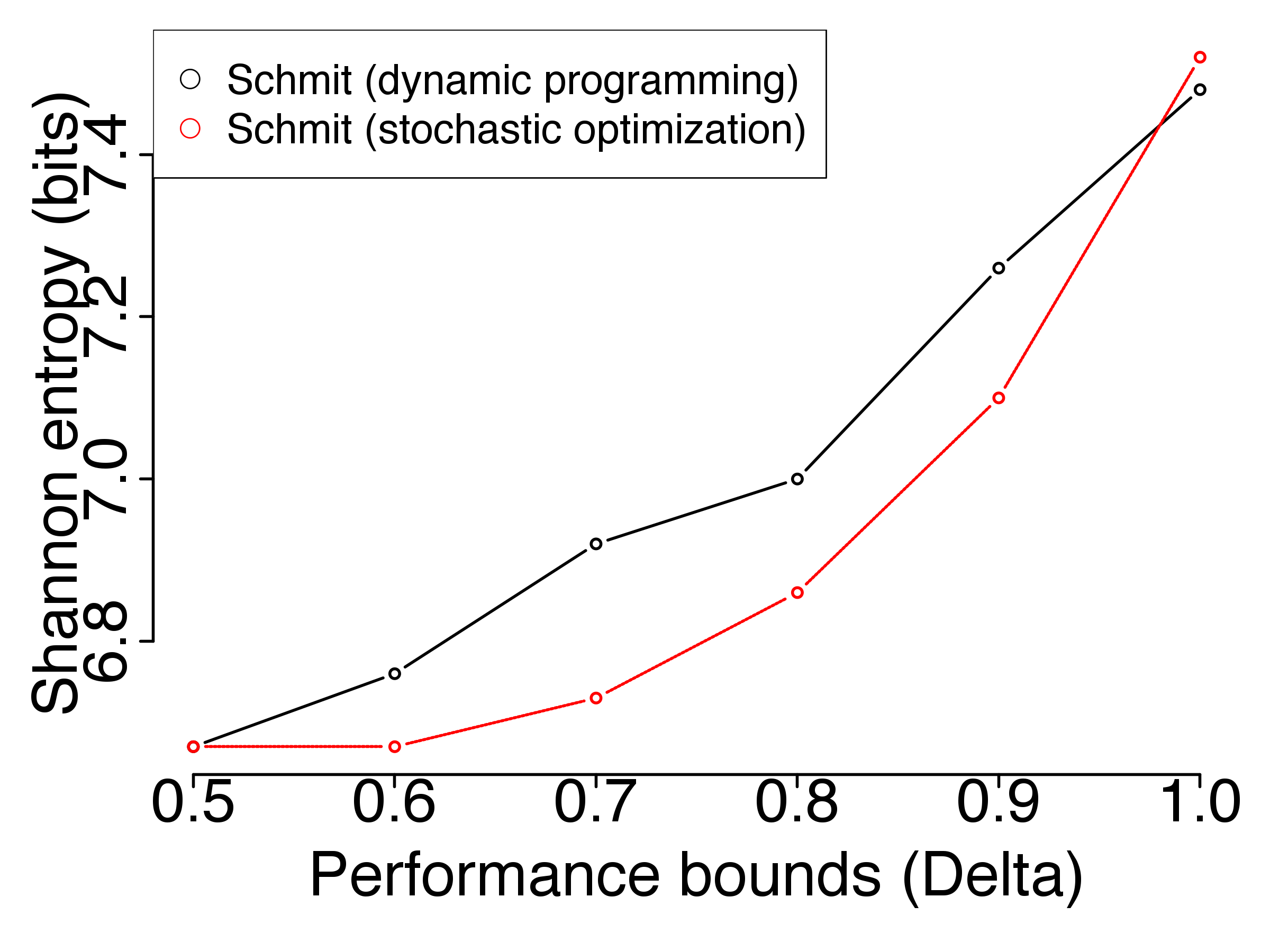}
	\includegraphics[width=0.32\textwidth]{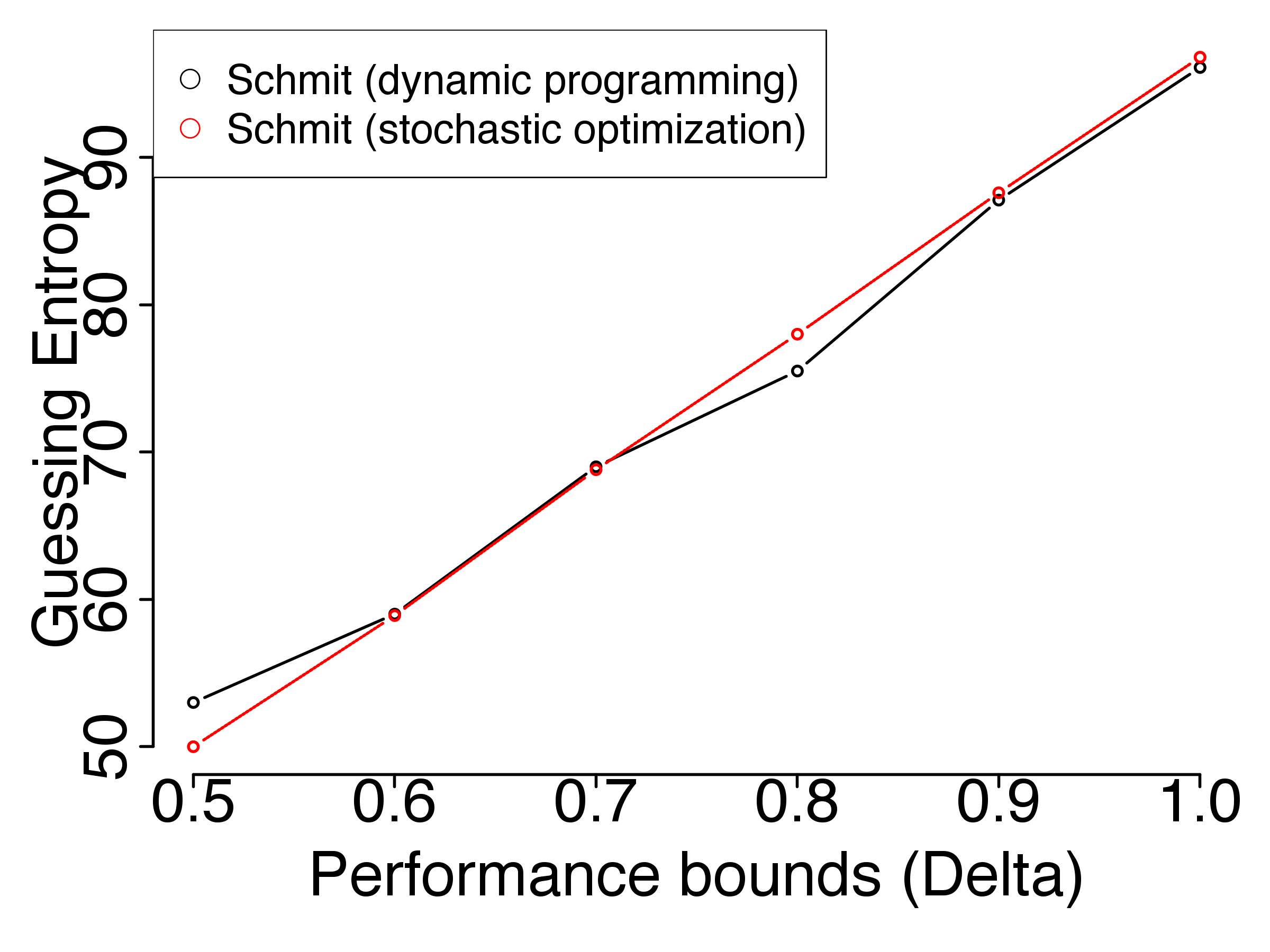}
	\caption{Entropy values versus performance overhead bounds on
		Branch\_and\_loop\_4.
		(a) Min-guess entropy, (b) Shannon Entropy, (c) Guessing entropy.
	}
	\label{fig:comp-entropy-delta}
\end{figure}

\subsection{Performance bounds versus entropy measures.}
Fig~\ref{fig:comp-entropy-delta} shows the relations between relaxing
the performance bound and the entropy values.
For min-guess entropy, Fig~\ref{fig:comp-entropy-delta}(a) shows
that the stochastic optimization improves the entropy gradually from 95
to 186 by relaxing the bound. However, the dynamic programming
has only improved when the performance bound exceeds 1.0.
For the Shannon and guessing entropy, Fig~\ref{fig:comp-entropy-delta}(b) shows
how \toolname improves the entropy with relaxing the performance bounds.